\def\aut{\mathop{\rm Aut}}
\def\rank{\mathop{\rm rank}}
\def\deg{\mathop{\mathrm{deg}}} 
\def\wgt{\mathop{\mathrm{wgt}}} 
\def\Ep{\mathbb{E}_{p}}
\def\Pp{\mathbb{P}_{p}}
\def\ker{\mathop{\rm ker}}
\def\im{\mathop{\rm im}}

\def\css{\mathop{\rm CSS}\nolimits}
\def\supp{\mathop{\rm supp}\nolimits}

\documentclass[aip,jmp,
  twocolumn,reprint,%
  tightenlines,letterpaper,%
  notitlepage,longbibliography]{revtex4-2}

  \usepackage{hyperref}

  \usepackage{dcolumn} 
  \usepackage{amssymb}

\usepackage{amsthm}
\usepackage{thmtools,thm-restate}

\usepackage[dvipsnames]{xcolor}
\usepackage{graphicx}
\graphicspath{{./}{./figs/}{./all-figs/}}

\def\grayout#1{{\color{gray}#1}}

\newtheorem{theorem}{Theorem}
\newtheorem{lemma}[theorem]{Lemma}
\newtheorem{conjecture}[theorem]{Conjecture}

\newtheorem{example}[theorem]{Example}

\usepackage[inline]{enumitem}
\setlist{nosep}

\usepackage{datetime}

\begin{document}

\title{Homology-changing percolation transitions on finite graphs}
\author{Michael Woolls}
\affiliation{Department of Physics \& Astronomy,
  University of California, Riverside, California 92521, USA}
\author{Leonid P. Pryadko}
\email{leonid.pryadko@ucr.edu}
\affiliation{Department of Physics \& Astronomy,
  University of California, Riverside, California 92521, USA}

\date\today

\begin{abstract}
  We consider homological edge percolation on a sequence
  $(\mathcal{G}_t)_t$ of finite graphs covered by an infinite
  (quasi)transitive graph $\mathcal{H}$, and weakly convergent to
  $\mathcal{H}$.  Namely, we use the covering maps to classify
  $1$-cycles on graphs $\mathcal{G}_t$ as homologically trivial or
  non-trivial, and define several thresholds associated with the rank
  of thus defined first homology group on the open subgraphs.  We
  identify the growth of the homological distance $d_t$, the smallest
  size of a non-trivial cycle on $\mathcal{G}_t$, as the main factor
  determining the location of homology-changing thresholds.  In
  particular, we show that the giant cycle erasure threshold $p_E^0$
  (related to the conventional erasure threshold for the corresponding
  sequence of generalized toric codes) coincides with the edge
  percolation threshold $p_{\rm c}(\mathcal{H})$ if the ratio
  $d_t/\ln n_t$ diverges, where $n_t$ is the number of edges of
  $\mathcal{G}_t$, and we give evidence that
  $p_E^0<p_{\rm c}(\mathcal{H})$ in several cases where this ratio
  remains bounded, which is necessarily the case if $\mathcal{H}$ is
  non-amenable.
\end{abstract}
\maketitle

\section{Introduction}
It is the threshold theorem\cite{Shor-FT-1996,Steane-FT-1997,
  Gottesman-FT-1998,Dennis-Kitaev-Landahl-Preskill-2002,
  Knill-FT-2003,*Knill-2004B,
  *Aliferis-Gottesman-Preskill-2006,*Reichardt-2009,
  Katzgraber-Bombin-MartinDelgado-2009} that makes large-scale quantum
computation feasible, at least in theory.  Related is the notion of
quantum channel capacity $R_Q$, such that for any rational $R<R_Q$,
there exists a quantum error correcting code (QECC) with rate $R$
which can be used to suppress the logical error probability to any
chosen (arbitrarily small) level, but not for $R>R_Q$.  Here the code
rate $R\equiv k/n$ is the ratio of the number $k$ of the logical
(encoded) qubits to the length $n$ of the code.  The precise value of
the capacity is not known for most quantum channels of interest,
except for the \emph{quantum erasure channel} with qubit erasure
probability $p$, in which case $R_Q=\min(0,1-2p)$, see
Ref.~\onlinecite{Bennett-DiVincenzo-Smolin-1997}.

In practice, it is often easier to deal with the threshold \emph{error
  probability} for a given family (infinite sequence) of QECCs with
certain asymptotic code rate $R$.  Depending on the nature of the
quantum channel in question, the threshold error probability may be
related to the location of a thermodynamical phase transition in
certain spin model associated with the codes.  In particular, for a
family of qubit toric codes on transitive graphs locally isomorphic to
a regular euclidean or hyperbolic tiling
$\mathcal{H}$ 
under independent
$Z$ Pauli errors, the decoding threshold is upper bounded by the
position of the multicritical point located at the Nishimori line of
the Ising model on
$\mathcal{H}$, see
Refs.~\onlinecite{Dennis-Kitaev-Landahl-Preskill-2002,%
  Kubica-etal-color-2017,Jiang-Kovalev-Dumer-Pryadko-2018}.  It is
widely believed that the two thresholds coincide, at least for the
euclidean tilings like the infinite square lattice and square-lattice
toric codes.  With a slightly more general model of independent
$X$/$Z$ Pauli errors, the threshold is the minimum of the
corresponding thresholds for each error type which can be computed
independently.

A special case is the relation between quantum erasure errors and
percolation\cite{Delfosse-Zemor-2010,Delfosse-Zemor-2012,%
  Delfosse-Zemor-2014}.  An erasure is formed by rendering inoperable
all qubits in a known randomly selected set.  Information loss happens
when erasure covers a logical operator of the code.  For certain code
families, and for qubit erasure probability
$p$ sufficiently small,
$p<p_E$, the probability to cover a codeword may go to zero as the
code length
$n$ is increased to infinity.  The corresponding threshold value
$p_E$ is called the \emph{erasure threshold} associated with the
chosen code family or code sequence.  With a
Calderbank-Shor-Steane (CSS)
code\cite{Calderbank-Shor-1996,Steane-1996}, one may consider the
erasure thresholds for $X$ and
$Z$ logical operators separately, so that the conventional erasure
threshold becomes $p_E=\min(p_E^{X},p_E^{Z})$.

The link between erasure and percolation thresholds is especially
simple in the case of toric/surface\cite{kitaev-anyons,%
  Bravyi-Kitaev-1998,Freedman-Meyer-1998,%
  Dennis-Kitaev-Landahl-Preskill-2002,Delfosse-Iyer-Poulin-2016} and
related quantum cycle codes\cite{Zemor-2009} where qubits are labeled
by the edges of a graph and, by convention, $Z$ logical operators are
supported on $1$-chains in certain equivalence classes, e.g., those
connecting two opposite boundaries of a rectangular region, or wrapped
around a torus.
Then, the erasure threshold $p_E^Z$ coincides with the discrete
version of the homological percolation
transition\cite{Bobrowski-Skraba-2020-euler,Bobrowski-Skraba-2020} for
$1$-chains.  It is also known that for square-lattice toric code the
erasure threshold $p_E^Z$ 
coincides\cite{Stace-Barrett-Doherty-2009,Fujii-Tokunaga-2012} with
the edge percolation threshold, $p_E^Z=p_{\rm c}(\mathbb{Z}^2)=1/2$.  On the
other hand, for a family of hyperbolic surface codes based on a given
infinite graph $\mathcal{H}$, a regular tiling on the hyperbolic
plane, we only know that the erasure threshold is upper
bounded\cite{Delfosse-Zemor-2010,Delfosse-Zemor-2012,Delfosse-Zemor-2014}
by the percolation threshold on ${\cal H}$, $p_E\le p_{\rm c}({\cal H})$.

Surely, the erasure and the percolation thresholds cannot always
coincide.  Indeed, percolation threshold is associated with the
formation of an infinite cluster; it is defined on an infinite graph,
while quantum codes are finite.  Further, erasure threshold is not a
bulk quantity, as it can be rendered zero by removing a vanishingly
small fraction of well-selected qubits.  Similarly, many different
finite graphs can be associated with a given infinite graph
${\cal H}$, and it is not at all clear that the erasure threshold
should remain the same independent of the details.

The goal of this work is to quantify the relation between edge
percolation and the stability of quantum cycle codes (QCCs) to erasure
errors.  Specifically, we consider sequences of finite graphs
$\mathcal{G}_t=(\mathcal{V}_t,\mathcal{E}_t)$, $t\in \mathbb{N}$, with
a common infinite covering graph $\mathcal{H}$, and use the covering
map $f_t:\mathcal{H}\to \mathcal{G}_t$ to identify homologically
non-trivial cycles on $\mathcal{G}_t$.  The distance
$d_t\equiv d_{Z,t}$ of the corresponding quantum code (the smallest
length of a non-trivial cycle) necessarily diverges with $t$ when the
sequence converges weakly to $\mathcal{H}$.  First, we show that it is
the scaling of $d_{t}$ with the logarithm of the code block length,
$n_t\equiv |\mathcal{E}_t|$, that determines the location of the
$Z$-erasure threshold, or the $1$-chain lower erasure threshold
$p_E^Z\equiv p_E^0$,  the point above which the probability of an
open homologically non-trivial $1$-cycle remains non-zero in the limit
of arbitrarily large graphs $\mathcal{G}_t$.  Roughly, with
sublogarithmic distance scaling, $d_{t}/\ln n_t\to 0$ as $t\to\infty$,
$p_E^0=0$.  On the other hand, with superlogarithmic distance
scaling, $d_{t}/\ln n_t\to \infty$, $p_E^0$ coincides with the edge
percolation threshold $p_{\rm c}(\mathcal{H})$, so that for
$p<p_{\rm c}(\mathcal{H})$, probability to find an open homologically
non-trivial $1$-cycle be asymptotically zero.  We also give an example
of a graph family with logarithmic distance scaling, where the
inequality in the upper bound is strict, $p_E^0< p_{\rm c}({\cal H})$,
and give numerical evidence that for some regular tilings of the
hyperbolic plane, erasure threshold is strictly below the percolation
threshold, $p_E^0<p_{\rm c}(\mathcal{H})$.

Second, the distance $d_t$ grows at most logarithmically with $n_t$
when $\mathcal{H}$ is non-amenable, which is also a necessary
requirement to have a finite asymptotic code rate $k_t/n_t\to R>0$,
where $k_t$ is the number of encoded qubits.  For such a graph
sequence, we define a pair of thermodynamical homological transitions,
$p_H^0$ and $p_H^1$, which characterize singularities in the
\emph{erasure rate}, asymptotic ratio of the expected homology rank of
the open subgraph and the number of edges $n_t$.  Namely, erasure rate
is zero for $p<p_H^0$, it saturates at
$R$ 
for $p>p_H^1$, and it takes intermediate values in the interval
$p_H^0<p<p_H^1$ (subsequence construction may be needed in this regime
to achieve convergence).  We prove that $p_H^1-p_H^0>R$, and, if
$\mathcal{H}$ and its dual, $\widetilde{\cal H}$, is a pair of
transitive planar graphs, we show that $p_H^0=p_{\rm c}(\mathcal{H})$ and
$p_H^1=1-p_{\rm c}(\widetilde{\cal H})$; the latter point coincides with the
uniqueness threshold $p_{\rm u}(\mathcal{H})$ on the original graph.  We
also conjecture that the two homological transitions coincide with the
percolation and the uniqueness thresholds, respectively, for any
non-amenable (quasi)transitive graph, $p_H^0=p_{\rm c}(\mathcal{H})$ and
$p_H^1=p_{\rm u}(\mathcal{H})$.

The outline of the paper is as follows.  In Sec.~\ref{sec:notations}
we give the necessary notations.  We present our analytical results in
Sec.~\ref{sec:analytical} and numerical results in
Sec.~\ref{sec:numeric}, with the proofs collected in the Appendix.
In section \ref{sec:conclusions} we give the conclusions and discuss
some related open questions.

\section{Definitions}
\label{sec:notations}

\subsection{Classical binary and quantum CSS codes}
A linear binary code with parameters $[n,k,d]$ is a vector space
$\mathcal{C}\subseteq \mathbb{F}_2^n$ of length-$n$ binary strings of
dimension $k$, where the minimum distance $d$ is the smallest Hamming
weight of a non-zero vector in $\mathcal{C}$.  Such a code
$\mathcal{C}\equiv \mathcal{C}_G$ can be specified in terms of a
generator matrix $G$ whose rows are the basis vectors, or in terms of
a parity check matrix $H$,
$\mathcal{C}\equiv \mathcal{C}_H^\perp=\{c\in\mathbb{F}_2^n:
Hc^T=0\}$, where $\mathcal{C}_H^\perp$ denotes the space dual
(orthogonal) to $\mathcal{C}_H$.  A generator matrix and a parity
check matrix of any length-$n$ code satisfy
\begin{equation}
  GH^T=0,\quad 
  \rank G+\rank H=n;
  \label{eq:dual-matrices}
\end{equation}
such matrices are called \emph{mutually dual}.

If $I\subset \{1,\ldots,n\}$ is a set of bit indices, for any vector
$b\in\mathbb{F}_2^{n}$, we denote $b[I]$ the corresponding
\emph{punctured} vector with positions outside of $I$ dropped.
Similarly, $G[I]$ (with columns outside of $I$ dropped) generates the
code $\mathcal{C}_G$ \emph{punctured} to $I$, denoted
${\cal C}_G[I]\equiv {\cal C}_{G[I]}$.  A \emph{shortened} code is
formed similarly, except by puncturing only the vectors supported
inside $I$,
$$
\text{$\mathcal{C}$ shortened to $I$} = \left\{c[I]: \;c\in
  \mathcal{C}\; \wedge\;\supp(c)\subseteq I\right\}.
$$
We use ${G}_I$ to denote a generating matrix of the code ${\cal C}_G$
shortened to $I$.  If $G$ and $H$ is a pair of mutually dual
binary matrices, see Eq.~(\ref{eq:dual-matrices}), then $H_I$ is
a parity check matrix of the punctured code $\mathcal{C}_{G}[I]$,
and\cite{MS-book}
\begin{equation}
  \label{eq:puncture-shortening-rank}
  \rank G[I]+\rank H_I=|I|, 
\end{equation}
i.e., matrices $G[I]$ and $H_I$ are  mutually dual.  In
addition, if $\overline I=\{1,2,\ldots,n\}\setminus I$ is the
complement of $I$, then
\begin{equation}
  \label{eq:dual-puncture-rank}
  \rank G[\overline I]+\rank G_I=\rank G.
\end{equation}

For the present purposes, it is sufficient that an $n$-qubit quantum
CSS code $\mathcal{Q}=\css(G_X,G_Z)$ can be specified in terms of two
$n$-column binary \emph{stabilizer generator matrices} with mutually
orthogonal rows, $G_XG_Z^T=0$.  It is isomorphic to a direct sum of
two quotient spaces, $\mathcal{Q}=\mathcal{Q}_X\oplus\mathcal{Q}_Z$,
where $\mathcal{Q}_X=\mathcal{C}_{G_Z}^\perp/\mathcal{C}_{G_X}$ and
$\mathcal{Q}_Z=\mathcal{C}_{G_X}^\perp/\mathcal{C}_{G_Z}$.  Vectors in
$\mathcal{Q}_X$ and $\mathcal{Q}_Z$, respectively, are also called
$X$- and $Z$-logical operators.  Explicitly, $\mathcal{Q}_X$ is formed
by vectors in $\mathcal{C}_{G_Z}^\perp$, with any two vectors that
differ by an element of $\mathcal{C}_{G_X}$ identified (notice that
$\mathcal{C}_{G_X}\subset\mathcal{C}_{G_Z}^\perp$).  Such a pair of
vectors $c'=c+\alpha G_X$ that differ by a linear combination of the
rows of $G_X$ are called mutually degenerate; we write $c'\simeq c$.
The second half of the code, $\mathcal{Q}_Z$, is defined similarly,
with the two generator matrices interchanged.  For such $Z$-like
vectors, the degeneracy is defined in terms of the rows of $G_Z$.

The distances $d_X$ and $d_Z$ of a CSS code are the minimum weights of
non-trivial vectors in $\mathcal{Q}_X$ and $\mathcal{Q}_Z$,
respectively, e.g.,
$d_X= \min\{\wgt c:c\in \mathcal{C}_{G_Z}^\perp\setminus
\mathcal{C}_{G_X}\}$.  Any minimum-weight codeword is always
\emph{irreducible}, that is, it cannot be written as a sum of two
vectors with disjoint supports, one of them being a
codeword\cite{Dumer-Kovalev-Pryadko-bnd-2015}.  The conventional
distance, the minimum weight of a logical operator in $\mathcal{Q}$,
is $d=\min(d_X,d_Z)$.  The dimension $k$ of a CSS code is the
dimension of the vector space $\mathcal{Q}_X$ (it is the same as the
dimension of $\mathcal{Q}_Z$), the number of linearly independent and
mutually non-degenerate vectors that can be used to form a basis of
$\mathcal{Q}_X$.  For a length-$n$ code with stabilizer generator
matrices $G_X$ and $G_Z$,
\begin{equation}
  \label{eq:CSS-k}
  k=n-\rank G_X-\rank G_Z.
\end{equation}
The parameters of a quantum CSS code are commonly written as
$[[n,k,(d_X,d_Z)]]$ or just $[[n,k,d]]$.

Any CSS code formed by matrices $G_X$ and $G_Z$ of respective
dimensions $r_X\times n$ and $r_Z\times n$ also defines a binary chain
complex with three non-trivial vector spaces,
\begin{equation}
  \mathcal{A}:\ldots \leftarrow\{0\}\stackrel{\partial_0}\leftarrow\mathcal{A}_0\stackrel{\partial_1}\leftarrow
  \mathcal{A}_1\stackrel{\partial_2}\leftarrow\mathcal{A}_2\stackrel{\partial_3}\leftarrow\{0\}\leftarrow \ldots,\label{eq:3-chain}
\end{equation}
where the spaces $\mathcal{A}_i=\mathbb{F}_2^{a_i}$ have dimensions
$a_0=r_X$, $a_1=n$, and $a_2=r_Z$, and the non-trivial boundary
operators are expressed in terms of the generator matrices
$\partial_1=G_X$, $\partial_2=G_Z^T$.  This guarantees the defining
property of a chain complex, $\partial_i\partial_{i+1}=0$,
$i\in\mathbb{Z}$.  Then, the code $\mathcal{Q}_Z$ is defined
identically to the first homology group
$H_1(\mathcal{A})= \ker(\partial_{1})/\im(\partial_{2})$, where
elements of $\im(\partial_{2})$ called \emph{cycles} are linear
combinations of the columns of $\partial_2={G}_Z^T$, while elements of
$\ker(\partial_1)$ called \emph{boundaries} are vectors orthogonal to
the rows of $\partial_1={G}_X$.  The other definitions also match.  In
particular, the dimension $k$ of the quantum code is the rank of the
first homology group, $k=\rank H_1(\mathcal{A})$, while the definition
of the homological distance $d_1(\mathcal{A})$ matches that of $d_Z$.
The other code, $\mathcal{Q}_X$, corresponds to the co-homology group defined
in the co-chain complex $\widetilde{\mathcal{A}}$ formed similarly but
with the two matrices interchanged.

Let us now consider the structure of the homology group where the
space ${\cal A}_1$ is restricted so that only components with indices
in the index set $I\subset \{1,2,\ldots,n\}$ be non-zero.
Respectively, the spaces $\ker\partial_1={\cal C}_{G_X}^\perp$ and
$\im \partial_2={\cal C}_{G_Z}$ should be replaced with the
corresponding reduced spaces.  The result is isomorphic to a chain
complex ${\cal A}'_I$ where the two boundary operators
are obtained by puncturing and shortening, respectively:
$\partial_1'=G_X[I]$ and $\partial_2'=(G_Z)_I^T$.  The
dimension of thus defined restricted homology group is given by%
\begin{equation}
  \label{eq:rank-H1-restricted}
 k_I'\equiv \rank H_1({\cal A}'_I)=|I|-\rank G_X[I]-\rank (G_Z)_I.
\end{equation}
Using Eq.~(\ref{eq:dual-puncture-rank}), we also
get\cite{Delfosse-Zemor-2012}
\begin{equation}
  \label{eq:rank-H1-restricted-two}
  k_I'=|I|-\rank G_Z-\rank G_X[I]+\rank G_Z[\overline I].
\end{equation}
The corresponding result for the rank $\tilde k_{I}'$ of the
restricted co-homology group can be found by exchanging the matrices
$G_X$ and $G_Z$; this gives the duality relation
\begin{equation}
  \label{eq:rank-H1-restricted-duality}
  k_I'+\tilde k_{\overline I}'=k.
\end{equation}

\subsection{Graphs, cycles, and cycle codes}

We consider only simple graphs with no loops or multiple edges.  A
graph ${\cal G}=({\cal V}, {\cal E})$ is specified by its sets of
vertices ${{\cal V}}\equiv {{\cal V}}_{\cal G}$, also called sites,
and edges ${{\cal E}}\equiv {{\cal E}}_{\cal G}$.
Each edge $e\in {\cal E}$ is a set of two vertices, $e=\{u,v\}$; it
can also be denoted with a wave, $u\sim v$.  For every vertex
$v\in {\cal V}$, its degree $\deg(v)$ is the number of edges that
include $v$.  An infinite graph $\mathcal{G}$ is called quasi-transitive if
there is a finite subset $\mathcal{V}_0\subset\mathcal{V}_{\cal G}$
of its vertices, such that for every vertex $v\in\mathcal{V}$ there is
an automorphism (symmetry) of $\mathcal{G}$ mapping $v$ to an element
of $\mathcal{V}_0$.  A transitive graph is a quasi-transitive graph
where the subset $\mathcal{V}_0$ of vertex classes contains only one
element.  All vertices in a transitive graph have the same degree.

We say that vertices $u$ and $v$ are connected on ${\cal G}$ if there
is a path $P\equiv P(u_0,u_\ell)$ between $u\equiv u_0$ and
$v\equiv u_\ell$, a set of edges which can be ordered and oriented to
form a \emph{walk}, a sequence of vertices starting with $u$ and
ending with $v$, with each directed edge in $P$ matching the
corresponding pair of neighboring vertices in the sequence,
\begin{equation}
  \label{eq:path}
  P(u_0,u_\ell)=   \{u_0\sim u_1,  u_1\sim
  u_2,\ldots, u_{\ell-1}
  \sim u_\ell\}\subseteq {\cal E}.
\end{equation}
We call such a path open if $u_0\neq u_\ell$, and closed otherwise.
The path is called self-avoiding (simple) if $u_i\neq u_j$ for any
$0\le i<j\le \ell$, except that $u_0$ and $u_\ell$  coincide if the
path is closed.
The length of the path is the number of edges in the set,
$\ell=|{P}|$.  The distance $d(u,v)$ between vertices $u$ and $v$ is
the smallest length of a path between them.  Given a vertex
$v\in {\cal V}$ and a natural $r\in\mathbb{N}$, a ball
$\mathcal{B}(v,r;\mathcal{G})$ is the subgraph of ${\cal G}$ induced
by the vertices $u\in {\cal V}$ such that $d(v,u)\le r$.

The edge boundary $\partial\mathcal{U}$ of a  set of
vertices $\mathcal{U}\subseteq \mathcal{V}$ is the set of edges
connecting $\mathcal{U}$ and its complement
$\overline {\cal U}\equiv {\cal V}\setminus {\cal U}$.  Given an
exponent $\alpha\le 1$, we define the isoperimetric constant of a
graph,
\begin{equation}
  \label{eq:edge-expander}
  b_\alpha=\inf_{\emptyset\neq \mathcal{U}\subsetneq\mathcal{V}, |{\cal U}|\neq\infty}
  {|\partial \mathcal{U}|\over
    \left[\min\left(\left|\mathcal{U}\right|, \left|\overline{\cal
            U}\right|\right)\right]^\alpha}.
\end{equation}
For an infinite graph, or a set of finite graphs that includes graphs
of arbitrarily large size, particularly important is the largest
$\alpha$ such that the corresponding $b_\alpha>0$.  Such a graph (or
graph family) is called an $\alpha$-expander; when $\alpha<1$, the
related parameter $\delta\equiv (1-\alpha)^{-1}$ is called the
isoperimetric dimension.  Isoperimetric dimension of any regular
$D$-dimensional lattice is $\delta=D$.  When $\alpha=1$, the
isoperimetric constant $b_1$ of a graph $\mathcal{G}$ is called its
Cheeger constant, $h(\mathcal{G})=b_1$.  An infinite graph with a
non-zero Cheeger constant is called non-amenable.

A set of edges $C\subseteq {\cal E}$ is called a cycle if the degree
of each vertex in the subgraph induced by $C$, ${\cal G}'=({\cal V},C)$, is
even.  The set of all cycles on a graph ${\cal G}$, with the symmetric
difference defined as
$A \oplus B\equiv (A\setminus B)\cup (B\setminus A)$ used as the group
operation, forms an abelian group, the \emph{cycle group} of
${\cal G}$, denoted $\mathcal{C}({\cal G})$.  Clearly, a closed path
is a cycle.  A simple cycle is a self-avoiding closed path.

A graph ${\cal H}$ is called a covering graph of ${\cal G}$ if there
is a function $f$ mapping ${\cal V}_{\cal H}$ onto
${\cal V}_{\cal G}$, such that an edge $(u, v)\in{{\cal E}}_{\cal H}$
is mapped to the edge
$\biglb(f(u), f(v)\bigrb)\in{{\cal E}}_{\cal G}$, with an additional
property that $f$ be invertible in the vicinity of each vertex, i.e.,
for a given vertex $u'\in {{\cal V}}_{\cal H}$ and an edge
$(f(u'), v)\in{{\cal E}}_{\cal G}$, there must be a unique edge
$(u', v')\in{{\cal E}}_{\cal H}$ such that $f(v')=v$.  As a result,
given a path $P$ connecting vertices $u$ and $v$ on ${\cal G}$ and a
vertex $u'\in {\cal V}_{\cal H}$ such that $f(u')=u$, there is a
unique path $P'$ on ${\cal H}$, the lift of $P$, such that $f$ maps
the sequence of vertices $u_0'\equiv u$, $u_1'$, $u_2'$, \ldots\ in
$P'$ to that in $P$.  To simplify the notations, we will in some cases
write a covering map as a map between the graphs,
$f:\mathcal{H}\to \mathcal{G}$.

A set of vertices $u'$ with the same covering map image $u$,
$f(u')=u$, is called the \emph{fiber} of $u$.  A lift of a closed path
starting and ending with $u$ is either a closed path, or an open path
connecting two different vertices in the fiber of $u$.  We call a
simple cycle on ${\cal G}$ homologically trivial if all its lifts are
simple cycles (of the same length).  A cycle on ${\cal G}$ is trivial
if it is a union of edge-disjoint homologically trivial simple cycles.
The set of trivial cycles on ${\cal G}$, with ``$\oplus$'' used for
group operation, is a subgroup of the cycle group on ${\cal G}$.  We
denote such a group $\mathcal{C}_0({\cal H};f)$.  The corresponding
group quotient,
$H_1(f)\equiv\mathcal{C}(\mathcal{G})/\mathcal{C}_0(\mathcal{H};f)$,
is the (first) homology group associated with the map $f$; its
elements are equivalence classes formed by sets of cycles whose
elements differ by an addition of a trivial cycle.  Namely, cycles $C$
and $C'$ are equivalent, $C'\simeq C$, if $C'=C\oplus C_0$, with
$C_0\in\mathcal{C}_0(\mathcal{H};f)$.

The cycle space of a graph $\mathcal{G}=(\mathcal{V},\mathcal{E})$
with $n=|\mathcal{E}|$ edges can be defined algebraically in terms of
the vertex-edge incidence matrix $J\equiv J_{\cal G}$.  Namely, it
is isomorphic to the binary code
$\mathcal{C}_{J}^\perp\subset\mathbb{F}_2^n$ whose parity check matrix
is the incidence matrix $J$,
$\mathcal{C}(\mathcal{G})\cong\mathcal{C}_J^\perp$.  On the other
hand, the code $\mathcal{C}_J$ generated by the incidence matrix is
isomorphic to the cut space of the graph.  Elements of the cut space
are edge boundaries $\partial \mathcal{U}$ of different partitions
defined by sets of vertices $\mathcal{U}\subset\mathcal{V}$.

In principle, any set $\mathcal{C}'\subset \mathcal{C}(\mathcal{G})$
of cycles on $\mathcal{G}$ can be used to construct a binary matrix
$K$ with the rows orthogonal to $J$, $JK^T=0$; the code
$\mathcal{C}_K\subset\mathbb{F}_2^n$ is isomorphic to the subspace of
the cycle space generated by elements of $\mathcal{C}'$.  In
particular, given the covering map $f:\mathcal{H}\to \mathcal{G}$,
such a matrix $K$ can be constructed using a basis set of
homologically trivial cycles $\mathcal{C}_0(\mathcal{H};f)$.  Thus,
such a covering map has a chain complex (\ref{eq:3-chain}) associated
with it, where $\mathcal{A}_0$, $\mathcal{A}_1$, and $\mathcal{A}_2$
are spaces generated by sets of vertices, edges, and homologically
trivial cycles, respectively.  In particular, the support $\supp(a)$
of any vector $a\in{\cal A}_1$ corresponds to a set of edges.  The
boundary operators are given by the constructed matrices
$\partial_1=J$, $\partial_2=K^T$.  Equivalently, the same matrices can
be used to define a stabilizer code $\css(J,K)$ with generators
$G_X=J$ and $G_Z=K$.  We will denote such a \emph{quantum cycle code}
associated with the covering map $f:{\cal H}\to {\cal G}$ as
$\mathcal{Q}(\mathcal{H};f)$.  The length of the code is
$n=|\mathcal{E}|$, the number of encoded qubits
$k=\rank H_1(\mathcal{A})$ is the rank of the first homology group
associated with covering map, and the distances $d_Z$, $d_X$,
respectively, are the homological distances $d_1(\mathcal{A})$,
$d_1(\widetilde{\cal A})$ associated with the chain complex
$\mathcal{A}$ and the co-chain complex $\widetilde{\cal A}$.

Given a graph $\mathcal{H}=(\mathcal{V},\mathcal{E})$ and a subgroup
$\Gamma$ of its automorphism group $\aut(\mathcal{H})$, consider the
partition of $\mathcal{V}$ induced by $\Gamma$, where a pair of
vertices $u$, $v$ are in the same class iff there is an element
$g\in\Gamma$ such that $g(u)=v$.  Then, the \emph{quotient graph}
$\mathcal{G}=\mathcal{H}/\Gamma$ has the vertex set given by the set
of vertex classes, with an edge between any two classes which contain
a pair of vertices connected by an edge from $\mathcal{E}$.  A graph
quotient $\mathcal{H}/\Gamma$ is covered by $\mathcal{H}$ if no
neighboring vertices fall into the same class.  When $\mathcal{H}$ is
infinite, a finite quotient graph $\mathcal{H}/\Gamma$ is obtained
if the subgroup $\Gamma$ has a finite index; in such a case
$\mathcal{H}$ must be quasitransitive.

\subsection{Percolation transitions}
\label{sec:percolation}

We only consider Bernoulli edge percolation, where each edge
$e\in {\cal E}$ of a graph ${\cal H}=({\cal V},{\cal E})$ is
independently labeled as open or closed, with probabilities $p$ and
$1-p$, respectively.  We are focusing on the subgraph $[{\cal H}]_p$
remaining after removal of all closed edges; connected components of
$[{\cal H}]_p$ are called clusters.  For a given $v\in {\cal V}$, the
cluster which contains $v$ is denoted
$\mathcal{K}_v\subseteq [{\cal H}]_p$.  If $\mathcal{K}_v$ is
infinite, for some $v$, we say that percolation occurs.

Three observables are usually associated with percolation: the
probability that vertex $v$ is in an infinite cluster,%
\begin{equation}
\theta_v\equiv
\theta_v({\cal H},p)=\Pp(|\mathcal{K}_v|=\infty),
\label{eq:theta-v}
\end{equation}
the
connectivity function,
\begin{equation}
  \label{eq:connectivity}
  \tau_{u,v}\equiv
  \tau_{u,v}({\cal H},p)=\Pp\bigl(u\in{\cal K}_v\bigr),
\end{equation}
the probability that vertices $u$ and $v$ are in the same cluster, and
the local cluster susceptibility,
\begin{equation}
\chi_v\equiv
\chi_v({\cal H},p)=\Ep(|\mathcal{K}_v|),
\label{eq:chi-v}
\end{equation}
the expected size of the cluster connected to $v$.  Equivalently,
cluster susceptibility can be defined as the sum of probabilities for
individual vertices to be in the same cluster as $v$, i.e., as a sum of
connectivities,
\begin{equation}
  \label{eq:chi-v-alt}
  \chi_v=\sum_{u\in{\cal V}}\tau_{v,u}.
\end{equation}
The critical propability $p_{\rm c}$, the percolation threshold, is
associated with the formation of an infinite cluster.  There is no
percolation, $\theta_v=0$, for $p<p_{\rm c}$, but $\theta_v > 0$ for
$p>p_{\rm c}$.  An equivalent definition is based on the existence of an
infinite cluster anywhere on $[{\cal H}]_p$: the probability of
finding such a cluster is zero at $p<p_{\rm c}$, and one at $p>p_{\rm c}$, see,
e.g., Theorem (1.11) in Ref.~\onlinecite{Grimmett-percolation-book}
(the same proof works for any infinite connected graph).

Similarly, the critical probability $p_T$ is associated with
divergence of site susceptibilities: $\chi_v$ is finite for $p<p_T$
but not for $p>p_T$.  Again, in a connected graph, this definition does not
depend on the choice of  $v\in{\cal V}$.  If percolation occurs (i.e.,
with probability $\theta_v>0$, $|{\cal K}_v|=\infty$), then clearly
$\chi_v=\infty$.  This implies $p_{\rm c}\ge p_T$.  The reverse is known to
be true for percolation on quasi-transitive
graphs\cite{Menshikov-1986,Menshikov-Sidorenko-eng-1987}:
$\chi_v=\infty$ can only happen inside or on the boundary of the
percolation phase.  Thus, for a quasi-transitive graph, $p_{\rm c}=p_T$.

An important question is the number of infinite clusters on
$[{\cal H}]_p$, in particular, whether an infinite cluster is unique.
For infinite quasi-transitive graphs, there are only three
possibilities: (\textbf{a}) almost surely there are no infinite
clusters; (\textbf{b}) there are infinitely many infinite clusters;
and (\textbf{c}) there is only one infinite
cluster\cite{Benjamini-Schramm-1996,%
  Haggstrom-Jonasson-2006,Hofstad-2010}.  A third critical
probability, $p_{\rm u}$, is associated with the number of infinite
clusters.  Most generally, we expect $p_T\le p_{\rm c}\le p_{\rm u}$.  For a
quasi-transitive graph, one has\cite{Hofstad-2010}
\begin{equation}
  \label{eq:thresholds}
  0<p_T=p_{\rm c}\le p_{\rm u}.
\end{equation}
Here, $p_{\rm u}$ is the uniqueness threshold, such that there can be only
one infinite cluster for $p>p_{\rm u}$, whereas for $p<p_{\rm u}$, the number of
infinite clusters may be zero, or infinite.  For an amenable
quasitransitive graph, $p_{\rm c}=p_{\rm u}$
\cite{Aizenman-Kesten-Newman-1987,Burton-1989,Kesten-2002}; it was
conjectured by Benjamini and Schramm\cite{Benjamini-Schramm-1996} that
$p_{\rm c}<p_{\rm u}$ for non-amenable quasi-transitive graphs.  Among other
examples, the conjecture has been recently verified for a large class
of Gromov-hyperbolic graphs\cite{Hutchcroft-hyp-2019}.

In order for the uniqueness threshold to be non-trivial, $p_{\rm u}<1$, the
graph $\mathcal{H}$ has to have only one end.  That is, it can not be
separated into two or more infinite components by removing a finite
number of edges.

In addition to uniqueness of the infinite cluster, the same threshold
$p_{\rm u}$ can be characterized in terms of the connectivity
function\cite{Tang-2019}.  Namely,
$\inf_{u,v\in \mathcal{V}}\tau_{u,v}(p)>0$ for $p>p_{\rm u}$ and it is
zero for $p<p_{\rm u}$.  Further, for planar transitive graphs, the
uniqueness threshold is related to the percolation threshold on the
dual graph,
\begin{equation}
  p_{\rm u}(\mathcal{H})=1-p_{\rm c}(\widetilde{\cal H}),
  \label{eq:planar-duality}
\end{equation}
see the proof of Theorem 7.1 in
Ref.~\onlinecite{Haggstrom-Jonasson-2006}.  In the case of planar
amenable graphs where $p_{\rm c}(\mathcal{H})=p_{\rm u}(\mathcal{H})$,
the duality (\ref{eq:planar-duality}) is between the two percolation
transitions\cite{Grimmett-percolation-book}.

\section{Homology-changing transitions}
\label{sec:analytical}
\subsection{Weakly converging sequences of graphs with a common cover}

Consider a finite graph
$\mathcal{G}=(\mathcal{V}_{\cal G},\mathcal{E}_{\cal G})$ covered by
an infinite graph $\mathcal{H}=(\mathcal{V},\mathcal{E})$.  While the
graph $\mathcal{H}$ needs not be quasi-transitive, the set of vertex
degrees of $\mathcal{H}$ is finite and matches that of $\mathcal{G}$;
in particular, the two graphs have the same maximal degree
$\Delta_{\rm max}$.  The covering map
$f:\mathcal{V}\to \mathcal{V}_{\cal G}$ also defines a quantum cycle
code $\mathcal{Q}(\mathcal{H};f)$ with parameters
$[[n,k,(d_{X},d_{Z})]]$, where $n=|\mathcal{E}_{\cal G}|$ the number
of edges in $\mathcal{G}$, and $k=\rank H_1(f)$ the dimension of the
first homology group associated with the map $f$.  We are particularly
interested in the case where the graphs $\mathcal{G}$ and
$\mathcal{H}$ look identically on some scale.  Formally, this is
formulated in terms of the injectivity radius, defined as the largest
integer $r_f$ such that the map $f$ is one-to-one in any ball
$\mathcal{B}(v,r_f;\mathcal{H})$.  Necessarily, for any covering map
$f$, the injectivity radius $r_f\ge1$.  We start by giving lower
bounds for the distances $d_X$, $d_Z$ in terms of the injectivity
radius.

First, an injectivity radius $r_f$ implies that no two vertices
located at distance $r_f$ or smaller from any vertex on $\mathcal{H}$
map to the same vertex on $\mathcal{G}$.  On the other hand, any simple cycle
$C\subset \mathcal{G}$ of length $\ell$ is for sure covered by a ball
of radius $r=\lceil \ell/2\rceil$ centered on a vertex in $C$.  This
gives (formal proofs are given in the Appendix):
\begin{restatable}{lemma}{dZbnd}
  \label{th:dZ-bnd}
  Consider a finite graph $\mathcal{G}$ covered by an infinite graph
  $\mathcal{H}$, with the injectivity radius $r_f$. Then the minimum
  weight $d_Z$ of a non-trivial cycle on $\mathcal{G}$ satisfies the
  inequality $ 2r_f+1\le d_Z\le 2r_f+3$.
\end{restatable}

Second, the minimum distance $d_{X}$ is the minimum size of a
homologically non-trivial \emph{co-cycle}, a set of edges on
$\mathcal{G}$ which has even overlap with any homologically trivial
cycle, but is not a cut of $\mathcal{G}$.  A lower bound for $d_{X}$
requires some additional assumptions:
\begin{restatable}{lemma}{dXbnd}
  \label{th:dX-bnd}
  Consider a finite graph $\mathcal{G}$ covered by an infinite
  one-ended graph $\mathcal{H}$, with the injectivity radius $r_f$.
  Assume that the cycle group of $\mathcal{H}$ can be generated by
  cycles of weight not exceeding $\omega\ge3$.  Then, the minimum
  weight of a non-trivial co-cycle on ${\cal G}$ satisfies the
  inequality $d_X> r_f/\omega$.
\end{restatable}

In addition, it will be important that for any covering map
$f:\mathcal{H}\to \mathcal{G}$, the vertices of $\mathcal{G}$ can be
lifted in such a way that they induce a connected subgraph of
$\mathcal{H}$, just as a square-lattice torus with periodic boundary
conditions becomes a rectangular piece of the square lattice after
cutting two rows of edges.
\begin{restatable}{lemma}{FlatSub}
  \label{th:flat-sub}
  Let $\mathcal{G}$ be a finite connected graph, $\mathcal{H}$ its
  cover with the covering map
  $f:\mathcal{V}\to \mathcal{V}_{\cal G}$ and the
  injectivity radius $r_f$.  For any $v'\in\mathcal{V}$ let
  $v\equiv f(v')\in\mathcal{V}_{\cal G}$ be its image.  Then there
  exists a set of vertices $\mathcal{V}_f\subset \mathcal{V}$
  which contains a unique representative from the fiber of every
  vertex of $\mathcal{V}_{\cal G}$, such that the subgraph
  $\mathcal{H}_f\subset \mathcal{H}$ induced by $\mathcal{V}_f$ be
  connected and contains the ball $\mathcal{B}(v',r_f;\mathcal{H})$.
\end{restatable}

In the following, we consider not a single graph $\mathcal{G}$, but a
sequence $(\mathcal{G}_t)_{t\in\mathbb{N}}$ of finite graphs
$\mathcal{G}_t=(\mathcal{V}_t,\mathcal{E}_t)$ sharing an infinite
connected covering graph $\mathcal{H}=(\mathcal{V},\mathcal{E})$, with
the covering maps $f_t:\mathcal{V}\to \mathcal{V}_t$.  If the
corresponding sequence of injectivity radii $r_t\equiv r_{f_t}$
diverges, we say that the sequence $({\cal G}_t)_t$ weakly converges
to $\mathcal{H}$.  Such a convergent sequence can be constructed, e.g., as a
sequence of finite quotients of the graph $\mathcal{H}$ with respect
to a sequence of subgroups of its symmetry group, which requires
$\mathcal{H}$ to be quasitransitive.  We do not know whether
quasitransitivity of $\mathcal{H}$ is necessary to have a sequence of
finite graphs covered by $\mathcal{H}$ and weakly convergent to
$\mathcal{H}$.  By this reason, in the following, we specify
(quasi)transitivity only when necessary for the corresponding proof.

Given a graph sequence with a common covering graph ${\cal H}$, we use
$\mathcal{Q}_t$ to denote the CSS code with parameters
$[[n_t,k_t,(d_{Xt},d_{Zt})]]$ associated with the covering map $f_t$.
We also denote the ``flattened'' subgraphs from Lemma
\ref{th:flat-sub} as
${\cal H}_t\equiv \mathcal{H}_{f_t}\subset \mathcal{H}$.  When the
sequence $(r_t)_t$ diverges, we can always construct a subsequence
$(t_s)_s$, $t_{s+1}>t_s$, such that the corresponding sequence of
graphs $(\mathcal{H}_{t_s})_s$ be increasing,
$\mathcal{H}_{t_{s+1}}\subsetneq \mathcal{H}_{t_s}$.  To this end, it
is sufficient to take $r_{t_{s+1}}>n_{t_s}$, regardless of the
particular spanning trees used in the construction of the graphs
$\mathcal{H}_t$.

\subsection{Homology erasure thresholds}

Coming back to percolation, let $H_1(f_t,p)$ denote the first homology
group formed by classes of homologically non-trivial cycles on the
open subgraph $[{\cal G}_t]_p$.  We will consider several observables
that quantify the changes in homology in the open subgraphs at large
$t$ as the probability $p$ is increased.  The first two, defined by
analogy with corresponding quantities for $1$-cycle proliferation in
continuum percolation\cite{Bobrowski-Skraba-2020}, are designed to
detect any changes in homology compared to the empty graphs at $p=0$,
and the graphs with all edges present at $p=1$.  Respectively, we
define the probability that a homologically non-trivial cycle exists
in the open subgraph,
\begin{equation}
  \label{eq:prob-E}
 {\bf P}_E(t,p)\equiv \Pp\biglb(\rank H_1(f_t,p)\neq 0\bigrb),
\end{equation}
and the probability that not all homologically non-trivial cycles are
covered in the open subgraph,
\begin{equation}
  \label{eq:prob-A}
 {\bf P}_A(t,p)\equiv \Pp\biglb(\rank H_1(f_t,p)\neq k_t\bigrb).
\end{equation}
Equivalently, ${\bf P}_A(t,p)$ is the probability that the open
subgraph at $\bar p=1-p$ covers a homologically non-trivial co-cycle.
In terms of the associated CSS code $\mathcal{Q}_t$, ${\bf P}_E(t,p)$
and ${\bf P}_A(t,1-p)$ are the erasure probabilities for a $Z$- and an
$X$-type codeword, respectively.  These quantities do not necessarily
characterize bulk phase(s), as they may be sensitive to the state of a
sublinear number of edges.

As $p$ is increasing from $0$ to $1$, ${\bf P}_E(t,p)$ is monotonously
increasing from $0$ to $1$ while ${\bf P}_A(t,p)$ is monotonously
decreasing from $1$ to $0$.  Thus, a version of the subsequence
construction can be used to ensure the existence of their $t\to\infty$
limits almost everywhere on the interval $p\in[0,1]$.  Instead, we
define the (lower) cycle erasure threshold for any given graph
sequence,
\begin{equation}
  \label{eq:pE0}
  p_E^0=\sup\left\{p\in [0,1]: \lim_{t\to\infty}{\bf P}_E(t,p)=0\right\}.
\end{equation}
 Because of monotonicity of
${\bf P}_E(t,p)$ as a function of $p$, a zero limit at some $p=p_0>0$
ensures the limit exists and remains the same everywhere on the
interval $p\in [0,p_0]$.  Further, the absence of convergence of the
sequence ${\bf P}_E(t,p)$ at some $p=p_1$ implies that the superior
and the inferior limits at $t\to\infty$ must be different, which, in
turn, implies the existence of a subsequence convergent to the
non-zero limit given by $\limsup_{t\to\infty}{\bf P}_E(t,p_1)>0$.

Similarly, we define the upper cycle erasure threshold,%
\begin{equation}
  \label{eq:pE1}
  p_E^1=\inf\left\{p\in [0,1]: \lim_{t\to\infty}{\bf P}_A(t,p)=0\right\},
\end{equation}
as the smallest $p$ such that open subgraphs preserve the full-rank
homology group with probability approaching one in the limit of the
sequence.

Existence of a homologically non-trivial cycle not covered by open
edges implies that closed edges must cover a conjugate codeword, a
non-trivial co-cycle.  The related threshold on an infinite graph can
be interpreted in terms of a transition dual to percolation,
proliferation of the boundaries at the complementary edge
configuration, with all closed edges replaced by open edges, and v.v.,
so that the open edge probability becomes $\bar p=1-p$.  On a locally
planar graph, like a tiling of a two-dimensional manifold, the dual
transition maps to the usual percolation on the dual graph.

We also notice that the usual erasure threshold $p_E$ for a family (or
a sequence) of quantum codes corresponds to a non-zero probability of
an \emph{erasure}, a configuration where a codeword is covered by
erased qubits.  For a CSS code, this implies a non-zero probability
that either an $X$- or a $Z$-type codeword be covered.  For codes
${\cal Q}_t$ associated with covering maps $f_t:{\cal H}\to{\cal G}_t$ in
the sequence $(\mathcal{G}_t)_{t\in\mathbb{N}}$, the conventional
erasure threshold can be found in terms of the thresholds for cycles
and co-cycles,
\begin{equation}
  \label{eq:pE}
  p_E=\min(p_E^0,1-p_E^1).
\end{equation}

The following lower bound constructed using a Peierls-style counting
argument is adapted from
Ref.~\onlinecite{Dumer-Kovalev-Pryadko-bnd-2015}:
\begin{restatable}{statement}{thPeierlsOne}
  \label{th:peierls-ineq}
  Consider a sequence of finite graphs
  $(\mathcal{G}_t)_{t\in\mathbb{N}}$ with a common covering graph
  $\mathcal{H}$.  Let $\Delta_{\rm max}$ be the maximum degree of
  $\mathcal{H}$, and assume that for some $t_0>0$, the injectivity
  radius $r_t$ associated with the maps
  $f_t:\mathcal{H}\to\mathcal{G}_t$ at $t\ge t_0$ scales at least
  logarithmically with the number of edges $n_t$, $r_t\ge A\ln n_t$,
  with some $A>0$.  The cycle erasure threshold for the corresponding
  sequence of CSS codes $(\mathcal{Q}_t)_{t\in\mathbb{N}}$ satisfies
  the lower bound $p_E^0\ge e^{-1/(2 A)}/(\Delta_{\rm max}-1)$.
\end{restatable}
It follows from the fact that $\mathcal{Q}_t=\css(J_t,K_t)$, where
$J_t$ is the vertex-edge incidence matrix of $\mathcal{G}_t$, with row
weights given by the vertex degrees, and Lemma \ref{th:dZ-bnd}.

We would like to ensure that the conventional erasure threshold
(\ref{eq:pE}) also be non-trivial, which requires that $p_E^1<1$.  To
construct such an upper bound, which becomes a lower bound in terms of
$\bar p=1-p$ in the dual representation, it is
sufficient\cite{Dumer-Kovalev-Pryadko-bnd-2015} that rows of the
trivial-cycle--edge adjacency matrix $K_t$ have bounded weights, and
that the distance $d_{Xt}$ diverges logarithmically or faster with
$n_t$.  Notice that here we do not rely on Lemma \ref{th:dX-bnd} which
gives a rather weak lower bound for the distance but, instead,
directly assume desired scaling of the minimum weight $ d_{Xt}$ of a
non-trivial co-cycle with $n_t$.  We have
\begin{restatable}{statement}{thPeierlsTwo}
  \label{th:peierls-ineq-two}
  Consider a sequence of finite graphs
  $(\mathcal{G}_t)_{t\in\mathbb{N}}$ with a common covering graph
  $\mathcal{H}$, with the cycle group ${\cal C}(\mathcal{H})$
  generated by cycles of weight not exceeding $\omega>1$.  Further,
  assume that the minimum weight $d_{Xt}\equiv d_X({\cal H};f_t)$ of a
  non-trivial co-cycle associated with the map
  $f_t:\mathcal{H}\to\mathcal{G}_t$ grows at least logarithmically
  with the number of edges $n_t$, $d_{Xt}\ge A'\ln n_t$, for
  sufficiently large $t\ge t_0'$ and some $A'>0$.  The upper erasure
  threshold for the corresponding sequence of CSS codes
  $(\mathcal{Q}_t)_{t\in\mathbb{N}}$ satisfies the bound
  $1-p_E^1\ge e^{-1/A'}/(\omega-1)$.
\end{restatable}
Let us now relate the cycle erasure threshold $p_E^0$ with the bulk
percolation threshold.  Most generally, it serves as an upper bound:
\begin{restatable}{theorem}{thErasurePercol}
  \label{th:pE-easy}
  Consider  a sequence of finite graphs
  $(\mathcal{G}_t)_{t\in\mathbb{N}}$ covered by an infinite graph
  $\mathcal{H}$.  Then, 
  $p_E^0\le p_{\rm c}(\mathcal{H})$.
\end{restatable}
This includes the case where the sequence of the injectivity radii
remains bounded (no weak convergence to ${\cal H})$, in which case,
obviously, $p_E^0=0$.  More precise results for $p_E^0$ are available
with additional assumptions, including the scaling of the injectivity
radius with the logarithm of the graph size:
\begin{restatable}{theorem}{pElogA}
  \label{th:pE-logA}
  Consider a sequence of finite transitive graphs
  $(\mathcal{G}_t)_{t\in\mathbb{N}}$ covered by an infinite graph
  $\mathcal{H}$.  If the homological distance $d_{Zt}$ scales
  sublogarithmically with graph size,
  $\displaystyle\lim_{t\to\infty}{d_{Zt}\over\ln n_{t}\strut}=0$, then
  $p_E^0=0$.
\end{restatable}

\begin{restatable}{theorem}{pElogB}
  \label{th:pE-logB}
  Consider a sequence of finite graphs
  $(\mathcal{G}_t)_{t\in\mathbb{N}}$ covered by an infinite
  quasi-transitive graph $\mathcal{H}$.  If the injectivity radius
  scales superlogarithmically with the graph size,
  $\displaystyle\lim_{t\to\infty}{r_t\over\ln n_t}=\infty$, then
  $p_E^0=p_{\rm c}$.
\end{restatable}

Information about the other threshold,
$p_E^1$, can be obtained in the planar case with the help of duality:

\begin{restatable}{corollary}{pEplanar}
  \label{th:pE-planar}
  Let ${\cal H}$ and $\widetilde{\cal H}$ be a pair of mutually dual
  infinite quasitransitive planar graphs.  Consider a sequence of
  finite graphs $({\cal G}_t)_{t\in\mathbb{N}}$ weakly convergent to
  ${\cal H}$, a cover of the graphs in the sequence.  Then,
  \begin{enumerate}[label={\em(\textbf{\roman*})}]
  \item
 $p_E^1\ge 1-p_{\rm c}(\widetilde{\cal
      H})$.      In addition, \smallskip
  \item if the graphs $\mathcal{G}_t$ in the sequence are transitive,
    $t\in\mathbb{N}$, and the injectivity radius grows
    sublogarithmically with the graph size, then $p_E^1=1$;
  \item if the injectivity radius grows superlogarithmically, then
    $p_E^1=1-p_{\rm c}(\widetilde{\cal H})$.
\end{enumerate}
\end{restatable}
Notice that for a superlogarithmic scaling of the injectivity radius,
the graph must be amenable, in which case
$p_{\rm u}(\mathcal{H})=p_{\rm c}(\mathcal{H})$.  We also believe that
under conditions of the Corollary, the duality gives
$p_{\rm u}(\mathcal{H})=1-p_{\rm c}(\widetilde{\cal H})$, see
Eq.~(\ref{eq:planar-duality}), although we only found the proof for
the case where the graph $\mathcal{H}$ is
transitive\cite{Haggstrom-Jonasson-2006}.  Whenever such a duality
relation holds, the upper cycle erasure threshold is bounded below by
the uniqueness threshold, $p_E^1\ge p_{\rm u}(\mathcal{H})$; with
superlogarithmic scaling of the injectivity radius, the sequence of
thresholds collapses to a single point,
$p_E^0=p_E^1=p_{\rm c}(\mathcal{H})=p_{\rm u}(\mathcal{H})$.

These results leave out an important case of percolation with
logarithmic distance scaling.  It is easy to see that logarithmic
distance scaling does not necessarily imply that
$p_E^0$ and $p_{\rm c}(\mathcal{H})$ be equal:
\begin{example}[Anisotropic square-lattice toric codes]
  \label{ex:anisotropic-square}
  Consider a sequence of tori
  $\mathcal{G}_t=\mathcal{T}_{L_x(t),L_y(t)}$ obtained from the
  infinite square lattice
  $\mathcal{H}$ by identifying the vertices at distances
  $L_x(t)$ and $L_y(t)$ along the edges in $x$ and
  $y$ directions, respectively.  For some
  $A>0$, consider the scaling $L_x(t)=t$,
  $L_y(t)=e^{t/A}/(2t)$.  This gives $d_{Zt}=t$ and
  $n_t=e^{t/A}$, so that $d_{Zt}=A\ln
  n_t$.  The cycle erasure threshold
  $p_E^0$ for this graph sequence satisfies $e^{-1/A}/3<p_E^0\le
  e^{-1/A}$.
\end{example}
The upper bound follows from considering
$L_y(t)$ independent non-trivial cycles of length
$t$, while the lower bound is given by Statement
\ref{th:peierls-ineq}.  In comparison, for edge percolation on
infinite square lattice, $p_{\rm c}=1/2$.

In addition to Example \ref{ex:anisotropic-square}, in
Sec.~\ref{sec:numeric} we give numerical evidence that
$p_E^0<p_{\rm c}(\mathcal{H})$ for several families of hyperbolic codes based
on regular $\{f,d\}$ tilings of the hyperbolic plane (here $2df>d+f$;
these are known to have a finite asymptotic rate $R=1-2/d-2/f$).

\subsection{Erasure rate thresholds}

Logarithmic scaling of the minimum distance $d_{Zt}$ associated with
the first homology group is the largest one may hope for in the
important case when the covering graph $\mathcal{H}$ is
non-amenable.  We specifically focus on the case of a graph sequence
with \emph{extensive homology rank} scaling, i.e., where the
associated codes have an asymptotically finite rate,
$R\equiv \lim_{t\to\infty}k_t/n_t>0$.  For such graph sequences, we
also consider the expected dimension of the erased subspace per edge,
or the \emph{erasure rate},
\begin{equation}
  \label{eq:mean-erased}
  {\bf R}_E(t,p)\equiv 
  n_t^{-1}\,
  \Ep\biglb(\rank H_1(f_t,p)\bigrb).
 \end{equation}
Analogous quantity was analyzed in detail by Delfosse and
Z{\'e}mor\cite{Delfosse-Zemor-2012}.  Unlike the probabilities
${\bf P}_E$ and ${\bf P}_A$, the erasure rate ${\bf R}_E$ is a bulk
quantity which can be used to define a thermodynamical transition in
the usual sense.  For any $t\in\mathbb{N}$, the erasure rate
${\bf R}_E(t,p)$ is a monotonously increasing function of
$p\in [0,1]$, bounded by the values at the ends of the interval,
\begin{equation}
   0\le {\bf R}_E(t,p)\le R_t\equiv k_t/n_t\le 1.
  \label{eq:bounds-RE}
\end{equation}

Let us now consider the thresholds associated with the erasure rate
(\ref{eq:mean-erased}).  We define the lower $p_H^0$ and the upper
$p_H^1$ critical points as the values of $p$ where ${\bf R}_E(t,p)$ in
the limit of large $t$ starts to deviate from $0$ and from $R$,
respectively:
\begin{eqnarray}
  \label{eq:pH0-def}
  p_H^0&=&\sup\{p\in[0,1]: \lim_{t\to\infty}{\bf R}_E(t,p)=0\},\\
  \label{eq:pH1-def}
  p_H^1&=&\inf\{p\in[0,1]: \lim_{t\to\infty}{\bf R}_E(t,p)=R\}.
\end{eqnarray}
We call these, respectively, the lower and the upper {\em
  homological\/} thresholds.  Evidently,
$p_E^0\le p_H^0\le p_H^1\le p_E^1$.  The critical point
$p_H^0$ was discussed in
Refs.~\onlinecite{Delfosse-Zemor-2010,Delfosse-Zemor-2012}.  Our first
result, an analogue of the corresponding inequality for the Ising
model, Eq.~(34) in Ref.~\onlinecite{Jiang-Kovalev-Dumer-Pryadko-2018},
 gives a lower bound on the gap between the two homological
thresholds:
\begin{restatable}{theorem}{pHdiff}
  \label{th:pH-diff}
  Consider a sequence of finite graphs
  $(\mathcal{G}_t)_{t\in\mathbb{N}}$ weakly convergent to an infinite
  graph $\mathcal{H}$, a cover of the graphs in the sequence, with
  rate-$R$ extensive homology rank.    Then there is a finite gap between the two homological
  thresholds,
  \begin{equation}
    \label{eq:pH-diff}
    p_H^1-p_H^0\ge R.
  \end{equation}
\end{restatable}
Second, we prove an ``easy'' inequality relating the lower homological
threshold with the percolation threshold on the covering graph:
\begin{restatable}{theorem}{pHeasybounds}
  \label{th:pH-easy-bounds}
  For a sequence of finite graphs $(\mathcal{G}_t)_{t\in\mathbb{N}}$
  weakly convergent to an infinite graph $\mathcal{H}$, a cover of the
  graphs in the sequence with extensive homology rank,
  $ p_{\rm c}(\mathcal{H})\le p_H^0 $.
\end{restatable}

The remaining analytical result is obtained with the help of the usual
duality between locally planar graphs, and is therefore limited to
planar graphs $\mathcal{H}$:
\begin{restatable}{theorem}{PHplanar}
  \label{th:pH-planar}
  Let ${\cal H}$ and $\widetilde{\cal H}$ be a pair of infinite
  mutually dual transitive planar graphs.  Consider a
  sequence of finite graphs $({\cal G}_t)_{t\in\mathbb{N}}$ weakly
  convergent to ${\cal H}$, a cover of the graphs in the sequence with
  extensive homology rank.  Then,
  \begin{equation}
    ({\bf i})\ p_H^0=p_{\rm c}(\mathcal{H}),\quad
    ({\bf ii})\  p_H^1= 1-p_{\rm c}(\widetilde{\cal H})=p_{\rm u}({\cal H}).
    \label{eq:pH-planar-equality}
  \end{equation}
\end{restatable}
This is an easy consequence of two previous results: the expression
\cite{Delfosse-Zemor-2012} for the expected homology rate in terms of
the average inverse cluster sizes on the graph and its dual, and the
exponential decay\cite{Antunovic-Veselic-2008,Hermon-Hutchcroft-2019}
of the size of finite clusters away from the percolation point on
transitive
graphs.  

Notice that in Theorem \ref{th:pH-planar}, the lower and the higher
homological thresholds, respectively, are actually associated with the
percolation and the uniqueness thresholds on the infinite graph
$\mathcal{H}$.  We believe this is not a coincidence, and put forward
\begin{conjecture}
  Consider a sequence of finite graphs $({\cal G}_t)_{t\in\mathbb{N}}$
  weakly convergent to a quasitransitive infinite  graph
  ${\cal H}$, a cover of the graphs in the sequence with extensive
  homology rank.  Then,
  \begin{equation}
    ({\bf i})\ p_H^0=p_{\rm c}(\mathcal{H}),\quad
    ({\bf ii})\  p_H^1=p_{\rm u}({\cal H}).
    \label{eq:pH-pc-pu-equality}
  \end{equation}
\end{conjecture}
Such a result makes sense, since neither the percolation nor the
uniqueness thresholds can be seen locally, by examining a finite
subgraph of ${\cal H}$.  Similarly, the homological transitions
require changes in cycles of length exceeding the injectivity radius,
which diverges without a bound.

\section{Numerical results for locally planar hyperbolic codes}
\label{sec:numeric}

In addition to analytical results, we also evaluated the erasure and
the percolation thresholds numerically for several families of planar
hyperbolic codes, as well as for a planar euclidean family of square
lattice toric codes.  Each family corresponds to a particular infinite
graph $\mathcal{H}_{f,d}$, regular tiling of the hyperbolic or
euclidean plane, parameterized by the Schl\"afli symbol $\{f,d\}$,
with $2/d+2/f\le 1$.  In such a graph, $d$ identical $f$-gons meet in
each vertex.  The finite graphs are
constructed\cite{Siran-2001,Sausset-Tarjus-2007} as finite quotients
of the corresponding graph $\mathcal{H}_{f,d}$ with respect to
subgroups of the symmetry group.

The parameters of the graphs used in the calculations are listed in
Tab.~\ref{tab:graphs}, where $\{f,d\}$ is the Schl\"afli symbol of the
corresponding tiling, $n$ is the number of edges, and $d_Z$ and $d_X$,
respectively, are the distances of the corresponding CSS codes.  The
smaller graphs with $n<10^3$ edges are from
N.~P.~Breuckmann\cite{Breuckmann-thesis-2017}.  We generated the
remaining graphs with a custom \texttt{GAP}\cite{GAP4} program, which
constructs coset tables of freely presented groups obtained from the
infinite van Dyck group $D(d,f,2)=\langle a,b|a^d,b^f,(ab)^2\rangle$
[here $a$ and $b$ are group generators, while the remaining arguments
are \emph{relators} which correspond to imposed conditions,
$a^d=b^f=(ab)^2=1$] by adding one more relator obtained as a pseudo
random string of generators to obtain a suitable finite group
${\cal D}$, a quotient of the original infinite group $D(d,f,2)$.
Then, the vertices, edges, and faces are enumerated by the right
cosets with respect to the subgroups $\langle a\rangle$,
$\langle ab\rangle$, and $\langle b\rangle$, respectively.  The
vertex-edge and face-edge incidence matrices $J$ and $K$ are obtained
from the coset tables.  Namely, non-zero matrix elements are in the
positions where the corresponding pair of cosets share an element.
Finally, the distance $d_Z$ of the CSS code $\css(J,K)$ was computed
using the covering set algorithm, which has the advantage of being
extremely fast when distance is
small\cite{Dumer-Kovalev-Pryadko-2014,Dumer-Kovalev-Pryadko-IEEE-2017},
and additionally verified by comparing the number of cycles through a
given vertex on the finite graph $\mathcal{G}$ and on a sufficiently
large subgraph of the infinite covering graph $\mathcal{H}_{f,d}$ (or
the corresponding dual graphs in the case of $d_X$).

\begin{turnpage}

\begin{table*}[hp]
  \resizebox{\textheight}{!}{
  \begin{tabular}{c|c|ccccccccccccccccccccccccccccccccc}
    &$n$ & 720 & 864 & 2448 & 6144 & 8640 & 18144 & 18216 & 19584 & 23760
    & 24360 & 24576 & 25308 & 25920 & 27360 & 29760 & 31200 & 32256
    & 32928 & 35280 & 36288 & 38880 & 40320 & 41040 & 46080 & 46656 \\
    \hline $\{4,6\}$& $d_Z$
        &\textbf{8} & \grayout{8} & \grayout{8} & \textbf{10} & \grayout{8}
&\grayout{10}&\grayout{10}&\grayout{8}&\grayout{10}&\grayout{10}&
10 & 10 & 10 & 10 & 10 & 10 & 10 &\grayout{8}& 10 & 10 & 10 & 10
                                            &\grayout{8}& 10 &\textbf{12}\\
$\{6,4\}$ &    $d_X$
        &\grayout{8}&\textbf{10}
                    &\textbf{12} &\grayout{12}&\textbf{14} & 12 & 14
                                                         & 14 & 12 & 14
            &\grayout{12}& 14 & 10 &\grayout{12}& 14 & 12 & 14 & 14
            &\grayout{10}&\grayout{14}&\grayout{12}& 14 & 10 & 12 & 12
    \\ \hline & $n$ &
46800 & 48576 & 50616 & 51888 & 52416 & 58800 & 62400 \\ \cline{1-9} $\{4,6\}$&$d_Z$ & 8 &
10 & 10 & 11 & 10 & 11 & 12\\ $\{6,4\}$&$d_X$ & 12 & 14 & 14 & 14 &\grayout{12}& 15
                                                 &\grayout{12}\\ \hline\hline
  &  $n$  & 504 & 648 & 768 & 864 & 1080 &
1224 & 1944 & 2016 & 2448 & 2592 & 3072 & 3240 & 4032 & 4320 & 5616 &
5832 & 6000 & 6072 & 6144 & 7344 & 14880 & 16848 & 18216 & 25944 &
32256\\ \cline{1-27} $\{3,8\}$&$d_Z$ &\textbf{6}
&\grayout{6}&\grayout{6}&\grayout{6}&\grayout{6}& 6 & 6 &\textbf{8} &
8 & 6 & 8 & 8 & 8 & 8 & 8 & 8 & 8 & 8 & 8 & 8 & 8
&\textbf{9}&\textbf{10} & 9 & 10\\ $\{8,3\}$&$d_X$ 
&\grayout{14}& 14 &\textbf{16} & 12 & 16 & 15
&\grayout{12}&\textbf{18}& 16
&\grayout{18}&\grayout{18}&\grayout{18}&\grayout{18}&\grayout{16}& 16
& 18 &\textbf{22}& 20 &\grayout{18}& 18 & 22 & 20 &\textbf{24}
                                                    &\grayout{21}&\grayout{24}\\
    \hline\hline
   & $n$ & 660 & 1800 & 1920 & 3420 & 4860 &
5760 & 7440 & 9600 & 10240 & 11520 & 12180 & 14880 & 17100 & 19200 &
23040 & 29400 & 34440 & 37500 & 38880 & 43200 & 57600 & 58240 & 58800
                                                    & 60900 & 61440\\
    \cline{1-27}$\{4,5\}$& $d_Z$
        &\grayout{8}&\textbf{10}
&\grayout{10}&\grayout{10}& \textbf{12} & 10 & 12 & 12 & 12 & 10 & 12
& 12 &\textbf{14} & 12 & 12 & 14 & 14 & 14 & 12 & 14 & 10 & 13 & 14 &
                                                                      12 & 12\\ $\{5,4\}$& $d_X$ 
        &\textbf{10} &\grayout{10}& \textbf{12}
&\textbf{14} &\grayout{12}& 10 & 14 & 14 & 14 & 10 & 14 & 14
&\textbf{16} & 14 & 15 &\textbf{17} & 16 & 16 & 16 & 16 & 10 & 15
                                            &\textbf{18} & 16 & 12\\
    \hline\hline
    \cline{1-25}
  &  $n$ & 900 & 4800 &
9600 & 9720 & 10800 & 11520 & 14400 & 15360 & 17220 & 18750 & 19440 &
19800 & 21600 & 29120 & 29400 & 30450 & 38880 & 40960 & 51330 & 52800
                            & 56730 & 58240\\ \cline{1-24} $\{5,5\}$&$d_Z,d_X$ 
        & \textbf{8} & \textbf{10} &
10 & 10 & 9 & 8 & 8 & \textbf{11} & 10 & 10 & 10 & 10 & 11 &
                                                             \textbf{12} & 10 & 12 & 10 & 12 & 12 & 10 & 11 & 12\\ \hline\hline
& $n$ & 546 & 672 & 4914 & 5376 & 6090 & 17220 & 19866\\ \cline{1-9}
$\{3,7\}$& $d_Z$ &\textbf{7}&\textbf{8} &\textbf{10}&\textbf{12}& 10 & 12 &
                                                                            12\\$\{7,3\}$& $d_X$ & 14&\textbf{16}&\textbf{24}& 24  &\grayout{20}&\textbf{28}& 26
  \end{tabular}%
  }
  \caption{Parameters of the hyperbolic graphs used to calculate
    critical $p$ values.  Given a Schl\"afli symbol
    $\{f,d\}$, finite graphs
    $\mathcal{G}$ and their dual graphs $\widetilde{\cal
      G}$ are parameterized by the number of edges
    $n$; they have $2n/d$ and
    $2n/f$ vertices, respectively, and the first homology groups of
    rank $k=2+(1-2/d-2/f)n$.  Distances $d_Z$ and
    $d_X$ are the lengths of the shortest homologically non-trivial
    cycles on $\mathcal{G}$ and $\widetilde{\cal
      G}$, respectively.  Numbers in bold indicate the smallest graph
    found with such a distance; only such graphs were used for
    calculating the cycle erasure threshold
    $p_E^0$, see Figs.~\ref{fig:pE-hyp55} and ~\ref{fig:pE-hyp73}.
    Percolation transition critical point $p_{\rm
      c}$ on the infinite hyperbolic graph $\mathcal{H}_{f,d}$ (same
    as the giant cluster transition) was calculated using all graphs
    in the corresponding family, with the exception of graphs whose
    distances are shown in gray.}
\label{tab:graphs}
\end{table*}
\end{turnpage}

To analyze percolation, we used a version of the Newman--Ziff
(NZ) Monte Carlo algorithm\cite{Newman-Ziff-alg-2001}.  The original
version of the algorithm simultaneously draws from a sequence of
canonical ensembles with $x=1$, $2$, $\ldots$ open edges, by starting
with all closed edges and randomly adding one open edge at a time,
with the acceleration due to a lower cost of statistics update.  To
find the rank $k'$ of the first homology group associated with the
open subgraph, we used the formula
\begin{equation}
  \label{eq:k-K-tK}
  k'=x-|\mathcal{V}|+|\mathcal{K}'|-|\overline{\mathcal{K}}'|+1,
\end{equation}
where $x\equiv |\mathcal{E}'|$ is the number of open edges, and
$|\mathcal{K}'|$ and $|\overline{\cal K}'|$, respectively, are the
numbers of connected components in the open subgraph of $\mathcal{G}$
and in the closed subgraph of the corresponding dual graph
$\widetilde{\cal G}$.  Eq.~(\ref{eq:k-K-tK}) is a consequence of
Eq.~(\ref{eq:rank-H1-restricted-two}).  It can also be derived with
the help of the cycle rank Euler formula and the fact that any trivial
open cycle is a cut for the corresponding dual graph.  Respectively,
in our version of the NZ algorithm, we simultaneously evolve a pair of
dual subgraphs, starting with all closed edges on $\mathcal{G}$ and
all open edges on $\widetilde{\cal G}$, and adding an open edge to
$\mathcal{G}'$ and removing the corresponding open edge from
$\widetilde{\cal G}'$ at each step.  In addition, for each set of
average quantities $A_x$ computed in the canonical ensemble with
$x\in\{0,\ldots,n\}$ edges open, we calculated the corresponding
grand-canonical quantity
\begin{equation}
 \label{eq:grand-canonical}
A_p=\sum_{x=0}^{n}{n\choose x }p^x (1-p)^{n-x}A_x.
\end{equation}
For the sake of numerical efficiency, we restricted the summation to
terms with $|x-p n|<M\sqrt{ n p(1-p)} $, with $M=10^2$.  We verified
that the results do not change when $M$ is increased by a factor of
two.  For each graph, we run $10^6$ Newman--Ziff sweeps and saved the
grand-canonical averages of observables for $10^3$ values of $p$ with
intervals of $\Delta p=10^{-3}$.  In addition, to get an independent
estimate of the errors, each threshold calculation was repeated three
times.

A de-facto standard way for estimating the erasure threshold $p_E^0$
is the crossing point method.  The method is based on the expectation
that the block error probability is asymptotically zero for any
$p<p_E^0$ and is equal to one for $p>p_E^0$, with the crossover region
small for large codes.  Respectively, when the erasure probability
found numerically for several graphs is plotted as a function of $p$,
the corresponding lines are expected to cross in a single point, which
is identified as the pseudothreshold.

This works well for codes with power-law distance scaling.  An example
is shown in Fig.~\ref{fig:PE-4-4}, where the homological error
probability (\ref{eq:prob-E}) for several square lattice toric codes
with parameters $[[2d^2,2,d]]$ and $d$ ranging from $60$ to $220$ is
plotted as a function of the open edge probability $p$.  Visually, a
beautiful crossing point close to $p_E^0=1/2$ is observed.  To find
the corresponding erasure pseudothreshold, the data was fitted
collectively with polynomials of $\xi\equiv p-p^0$.  The polynomials
had different coefficients for different graphs, except the zeroth
order coefficient used to find the ordinate of the crossing point.
With the fit range $0.49 \le p\le 0.51$, the degree of the polynomials
was adjusted by hand to minimize the standard deviation of $p^0$, the
abscissa of the crossing point extracted from the data.  For the
square-lattice graphs, using 6th degree polynomials, we obtained
$p^0(\{4,4\})=0.500004\pm 0.000002$, very close to the square lattice
percolation threshold $p_{\rm c}(\{4,4\})=1/2$, as expected from
Refs.~\onlinecite{Stace-Barrett-Doherty-2009,Fujii-Tokunaga-2012} and
Theorem \ref{th:pE-logB}.  The corresponding linear terms $A_{1n}$ (the
derivative at the crossing point) have a power law $A_{1n}=b n^\alpha$
scaling (not shown), with the exponent $\alpha=0.375\pm 0.003$,
consistent with the expectation of a sharp threshold in the
large-$n$ limit.

\begin{figure}[htbp]
  \centering
  \includegraphics[width=0.6\columnwidth]{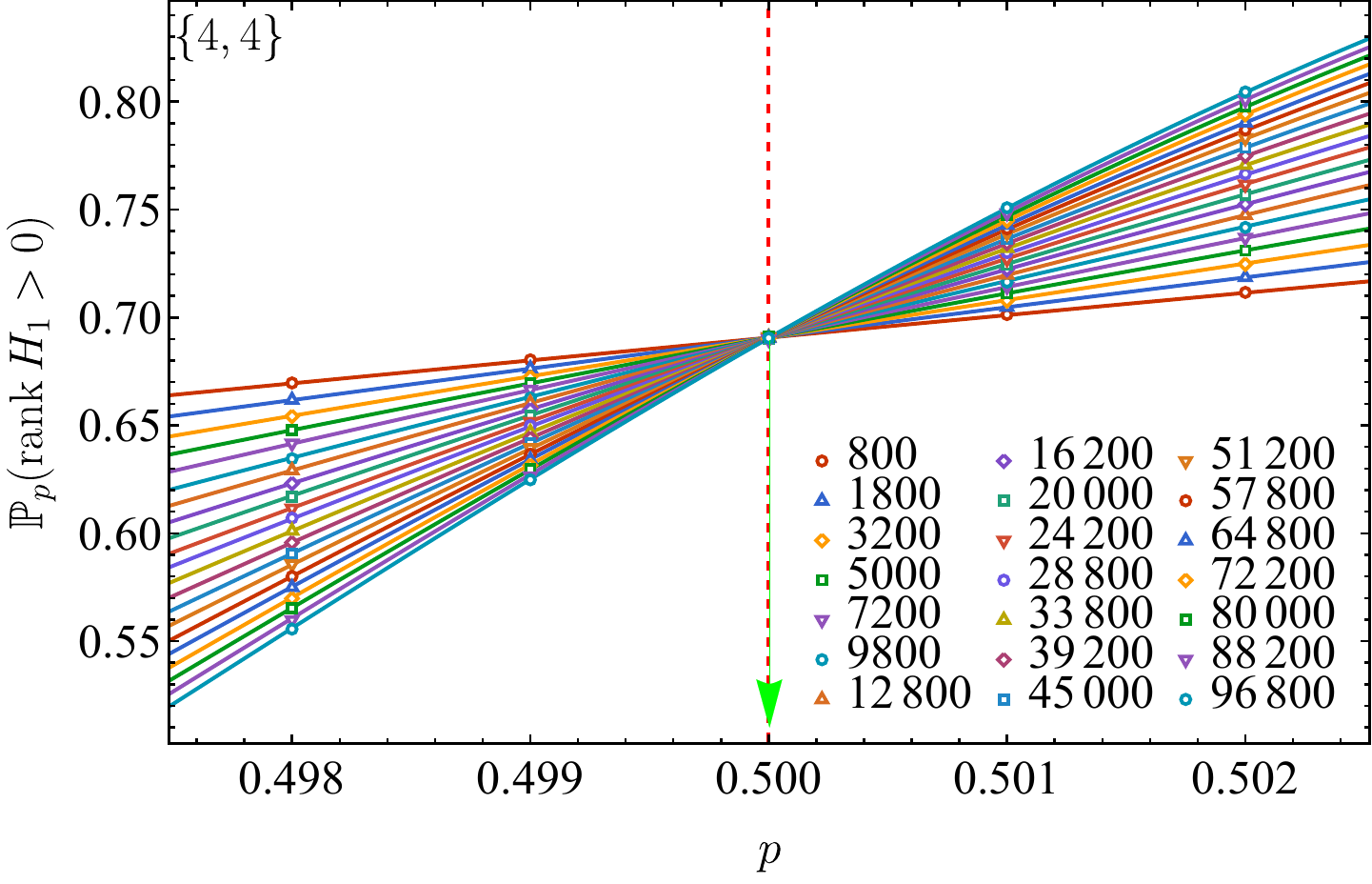}
  \caption{(Color online) Finding the erasure pseudothreshold for
    square lattice toric codes.  Symbols show the homological error
    probability (\ref{eq:prob-E}) evaluated numerically for different
    graphs labeled by the number of edges $n=2d^2$, with distances $d$
    ranging from $60$ to $220$, plotted as a function of open edge
    probability $p$.  As expected, beautiful crossing point very close
    to $p_E^0=1/2$ is observed.  Lines are the polynomials
    $f_n(\xi)=A_0+A_{1n}\xi+A_{2n}\xi^2+\ldots$ of $\xi=p-p^0$ of
    degree 6 obtained by fitting the data collectively in the range
    $0.49\le p\le 0.51$.  The vertical dashed line indicates the
    square lattice percolation threshold $p_{\rm c}=1/2$. }
\label{fig:PE-4-4}
\end{figure}

We used a similar technique to process the homological error
probability data for hyperbolic graphs.  A sample of the corresponding
plots is shown in the top portions of Figs.~\ref{fig:pE-hyp55} and
\ref{fig:pE-hyp73}.  These plots have two significant differences with
that in Fig.~\ref{fig:PE-4-4}.  First, the crossing points are
significantly below the percolation transitions indicated by the
vertical dashed lines.  Second, despite smaller scales, the
convergence near the crossing crossing points does not look as nice.
Empirically, deviations in the position of the curves are associated
with the differences in the ratio $\ln n/d$, cf.~the bounds in
Statements \ref{th:peierls-ineq}, \ref{th:peierls-ineq-two} and
Example \ref{ex:anisotropic-square}.  To reduce the corresponding
errors, in the calculation of the erasure thresholds we only used the
``optimal'' graphs, the smallest graphs with the corresponding
distances; such graphs are indicated in Tab.~\ref{tab:graphs} with the
distance shown in bold.

Yet, using only the optimal graphs was not sufficient to completely
eliminate the finite-size variation.  Much better crossing points are
obtained by introducing a vertical shift $B\, \ln n/d$, where $B$ is
an additional global fit parameter (see bottom plots in
Figs.~\ref{fig:pE-hyp55} and \ref{fig:pE-hyp73}).

\begin{figure}[htbp]
  \centering
  ~~\,\includegraphics[width=0.6\columnwidth]{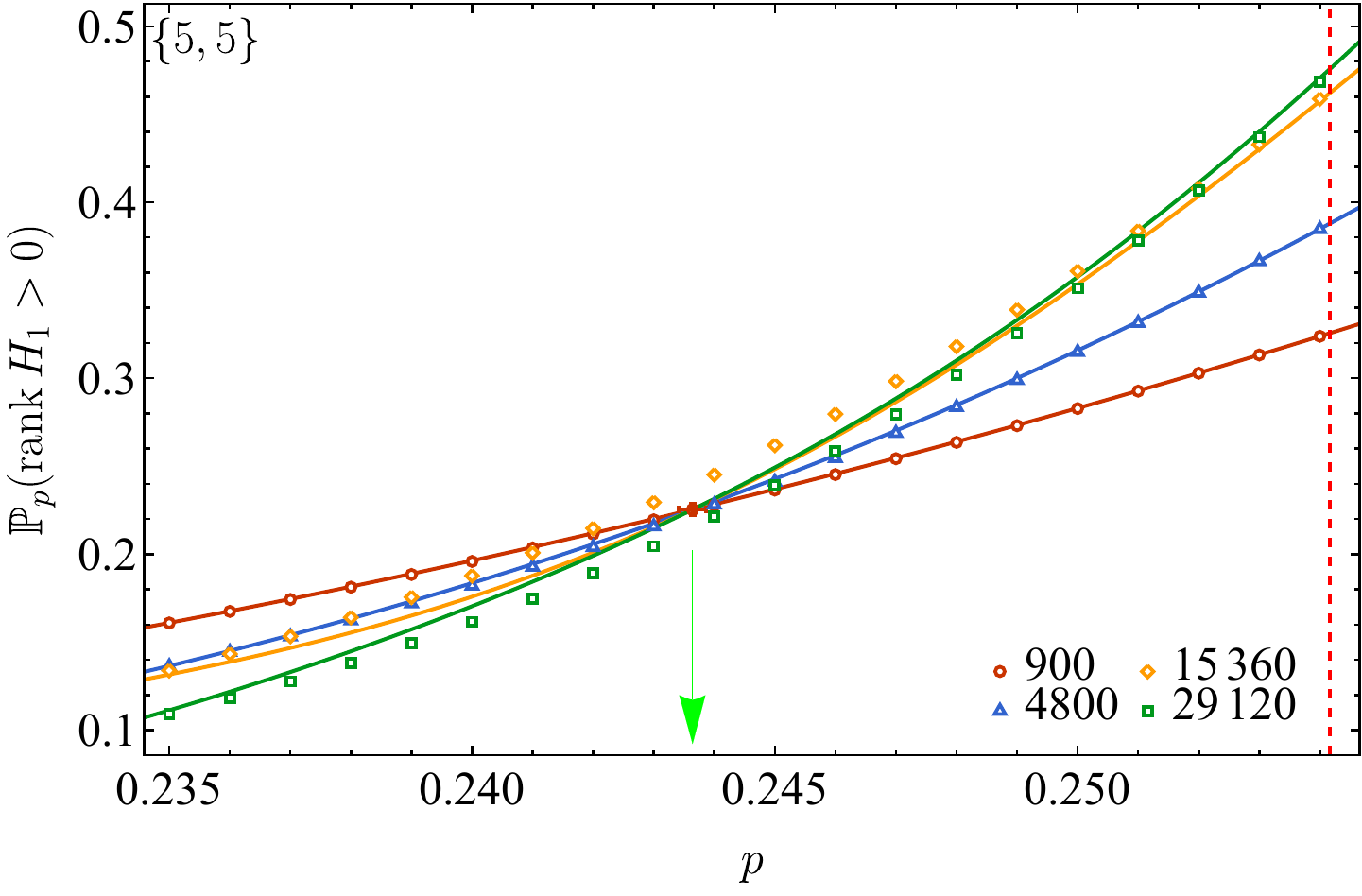}\\[0.15em]
  \includegraphics[width=0.62\columnwidth]{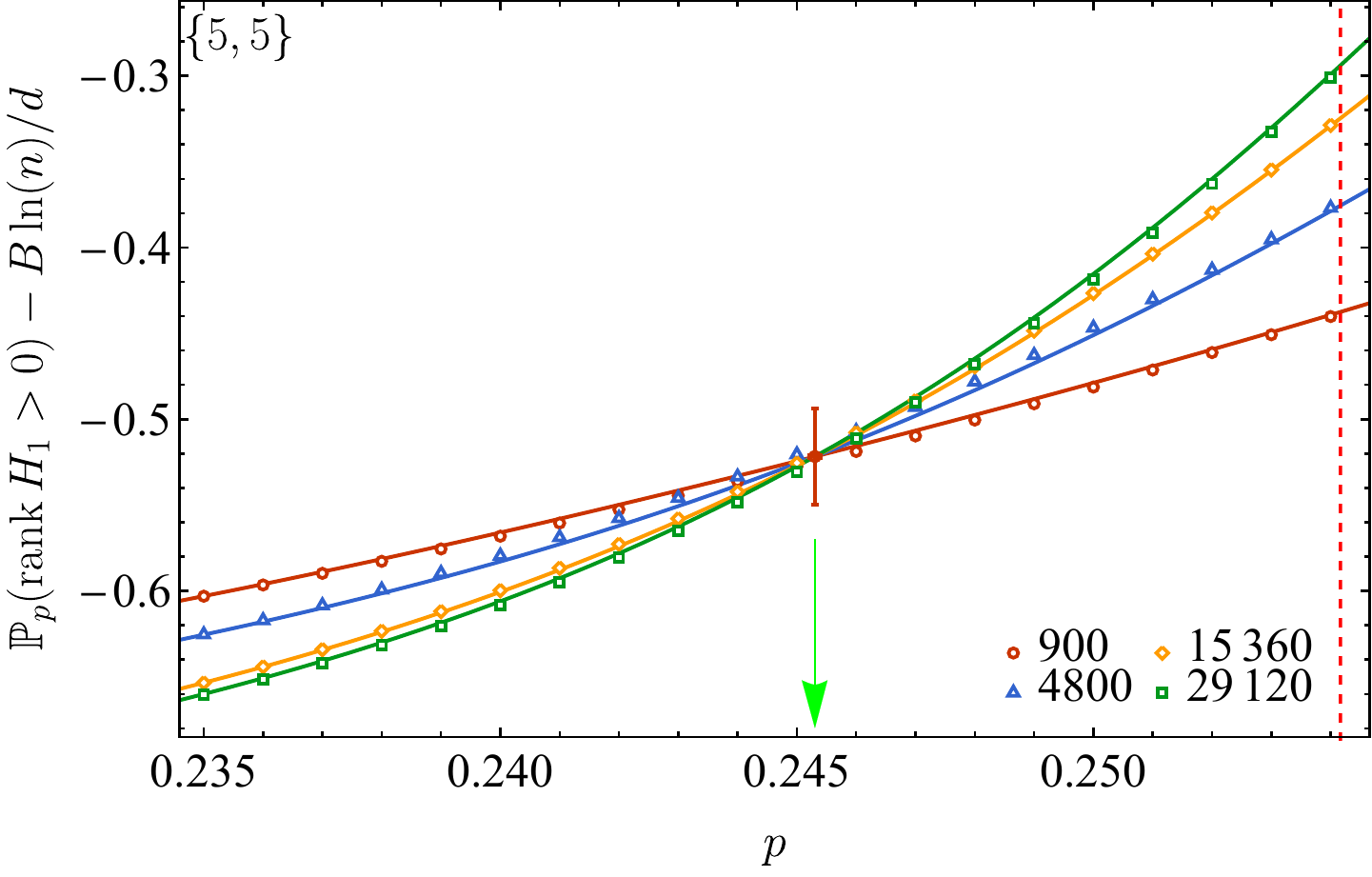}
  \caption{(Color online) Top: as in Fig.~\ref{fig:PE-4-4} but
    for the hyperbolic code family $\{5,5\}$.  The green arrow
    indicates the position of the crossing point found by the fit; it
    is significantly below the percolation threshold for the
    corresponding infinite lattice (vertical red dashed line).  In
    addition, the data for the graph with $n=15\,350$ is shifted
    upward, which we associate with a slightly smaller ratio
    $d/\ln n$, see Fig.~\ref{fig:dist}.  This is verified in the
    bottom plot, where an additional vertical shift proportional to
    $\ln n/d$ is added, which substantially improves the convergence
    at the crossing point.}
\label{fig:pE-hyp55}
\end{figure}

\begin{figure}[htbp]
  \centering
  \includegraphics[width=0.6\columnwidth]{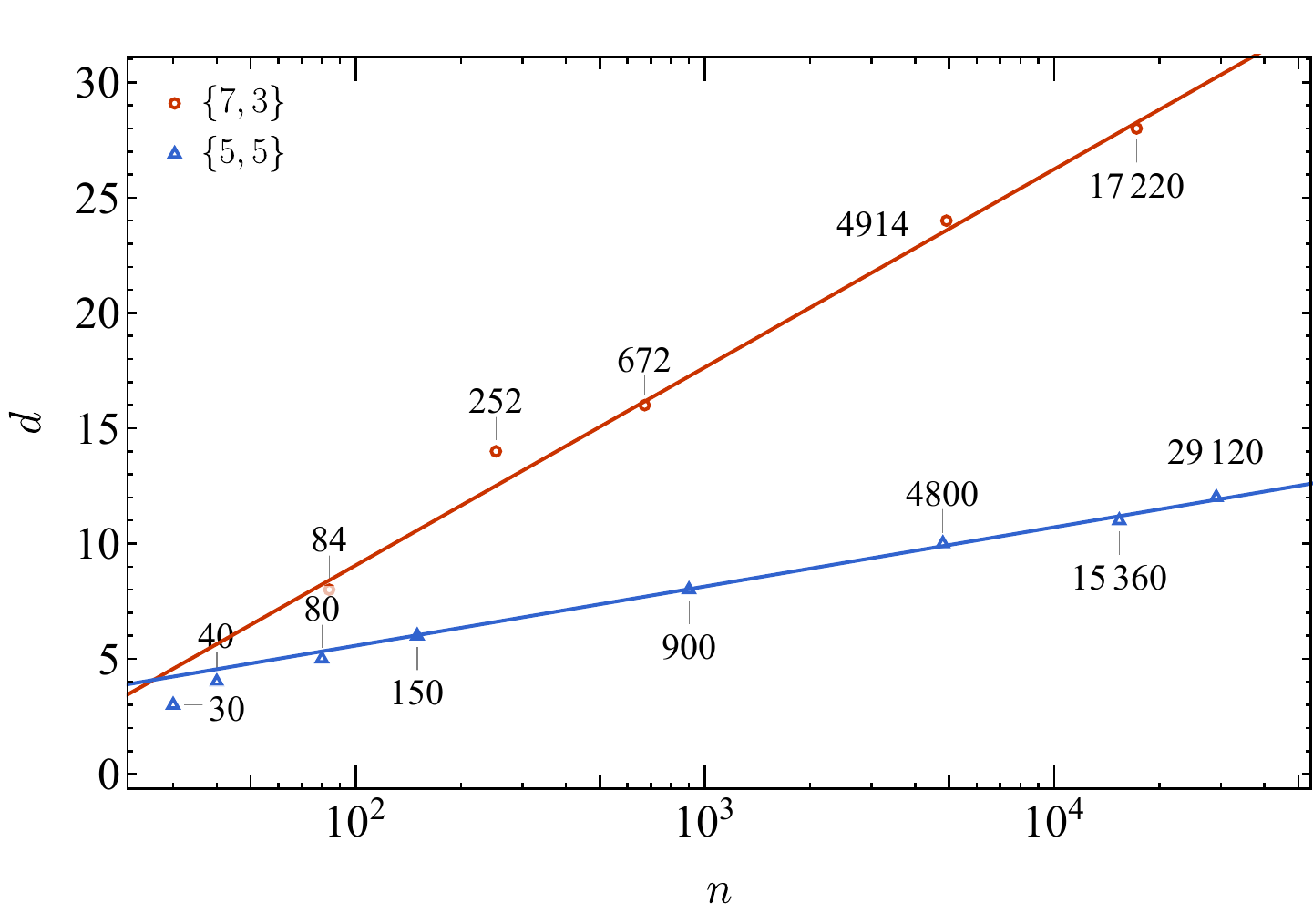}
  \caption{Homological distance $d\equiv d_Z$ associated with
    non-trivial cycles for optimal graphs in $\{7,3\}$ and $\{5,5\}$
    families vs.\ the graph size $n$ (number of edges) with the
    logarithmic scale.  Numbers also indicate the graph sizes.
    Smaller relative distances $d/\ln n$ result in larger erasure
    probabilities in Figs.~\ref{fig:pE-hyp55} and \ref{fig:pE-hyp73}
    (top); this can be compensated to some extend by using the
    correction term as in bottom plots in Figs.~\ref{fig:pE-hyp55} and
    \ref{fig:pE-hyp73}.}
\label{fig:dist}
\end{figure}

\begin{figure}[htbp]
  \centering
 \includegraphics[width=0.6\columnwidth]{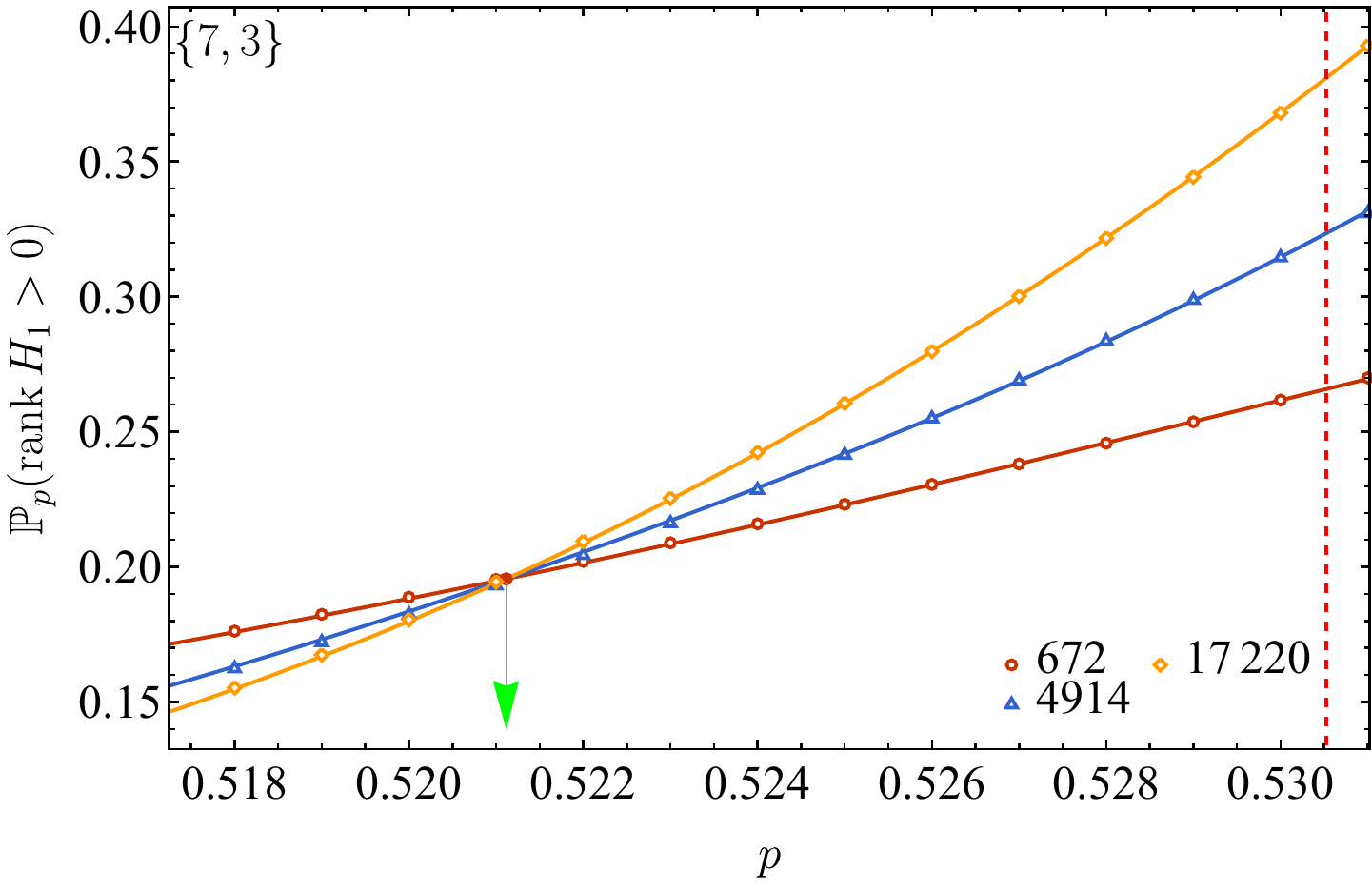}\\[0.15em]
  \includegraphics[width=0.6\columnwidth]{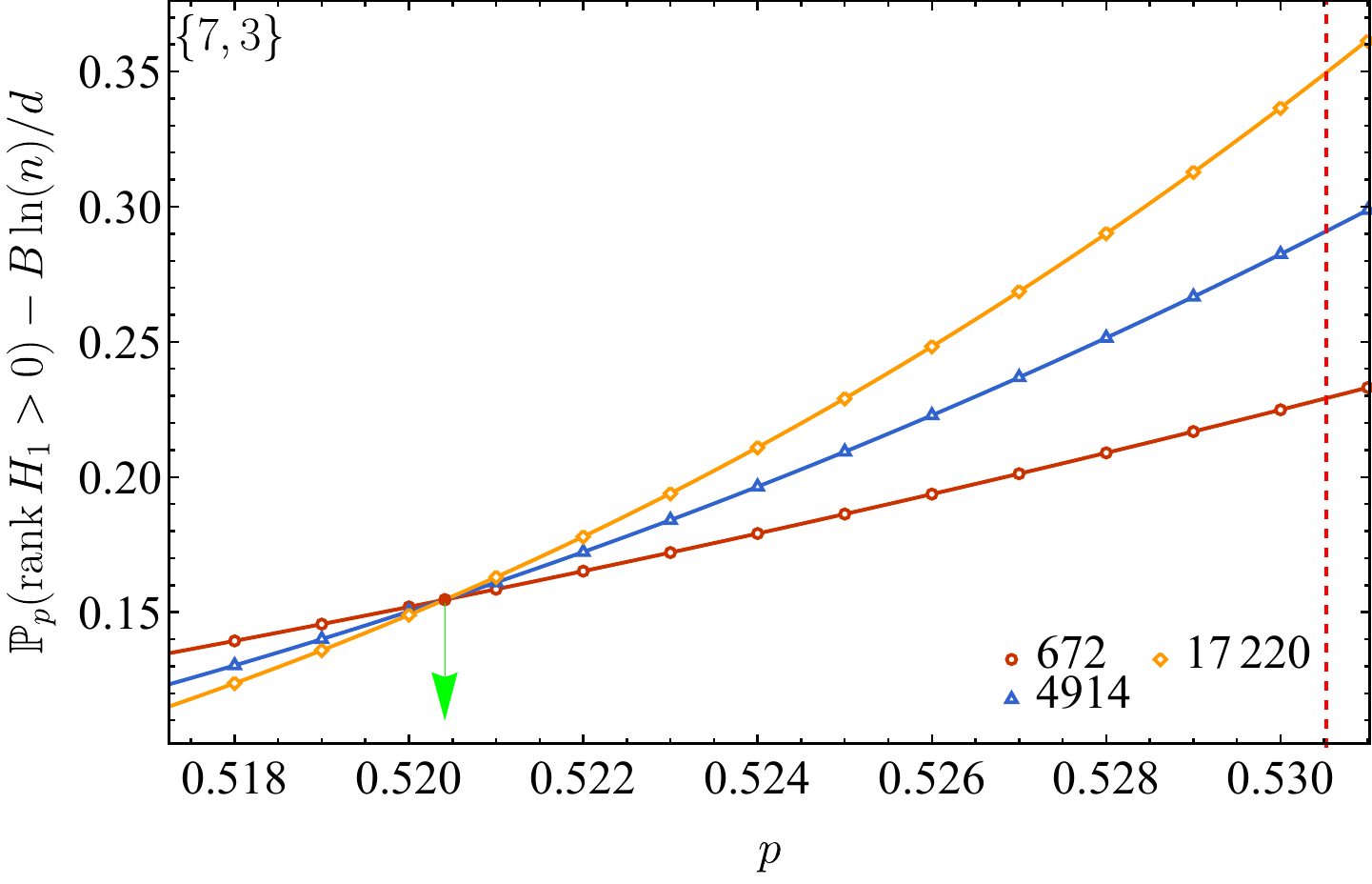}
  \caption{(Color online) As in Fig.~\ref{fig:pE-hyp55} but for the
    hyperbolic code family $\{7,3\}$.  The convergence at the crossing
    point is much better than that in Fig.~\ref{fig:pE-hyp55} (top),
    which we associate with substantially higher ratios $d/\ln n$ for
    the graphs in this family.  Bottom plot: addition of the
    additional vertical shift $B\ln n/d$ causes a substantial shift of
    the crossing point position without visibly improving the
    convergence.}
\label{fig:pE-hyp73}
\end{figure}

In comparison, the crossing point method does not work for measuring
the location of the homological transition $p_H^0$, even though the
variation between the graphs is not expected to matter that much here.
Main reason for the difference is that the erasure rate
(\ref{eq:mean-erased}) retains a finite slope in the infinite graph
limit, which makes the crossing point analysis unreliable.

To check for spurious errors, in our simulations we have also measured
the conventional percolation characteristics, in particular, the
average sizes $S_j=\langle {\cal K}_j\rangle $, $j=1,2,3$ of the three
largest clusters.  We have used several finite-size scaling techniques
to extract the location of the percolation transition which coincides
with the giant-cluster transition, see Theorem 1.3 in
Ref.~\onlinecite{Benjamini-Nachmias-Peres-2011}.  All techniques,
including the cluster-size ratio
technique\cite{Margolina-Herrmann-Stauffer-1982,daSilva-Lyra-Viswanathan-2002},
give transition points in a reasonable agreement with the values
expected from invasion percolation simulations in
Ref.~\onlinecite{Mertens-Moore-2017}.  For hyperbolic graphs
$\mathcal{H}_{f,d}$ with $(f+d)/fd<2$, we found that the most accurate
values of $p_{\rm c}$ are found using the technique based on the
expectation of cluster size scaling similar to that for random
graphs\cite{Kozma-Nachmias-2009,Heydenreich-vanderHofstad-book-2017},
$S_j\propto n^{2/3}$ near $p_{\rm c}$, with the critical region of
width $\Delta p\sim n^{-1/3}$.  Respectively, when interpolated values
of $p$ such that the expected size of the largest cluster satisfies
$S_1(p)=\omega n^{2/3}$ are plotted for $\omega\in\{1/4,1/2,1\}$ as a
function of $x\equiv n^{-1/3}$, the data for graphs with different $n$
fit nicely, and can be extrapolated to $x=0$ (infinite graph size)
using polynomial fits, see Fig.~\ref{fig:ns23}.  Notice that while
this technique works well for hyperbolic graphs and for random graphs,
in our simulations it failed dramatically for the planar $\{4,4\}$
graph family, as can be seen from the corresponding value of
$p_c^{(2/3)}$ in Tab.~\ref{tab:pc}.

\begin{table*}[htp]
  \centering
    \resizebox{\textwidth}{!}{
\begin{tabular}{r|l|ldl|ldl|ldl}
  $\{f,d\}$  & $p_{\rm c}$
  &$p_c^{(2/3)}$&\text{$n_\sigma^{(2/3)}$}&deg&$p_E^{(C)}$&\text{$n_\sigma^{(C)}$}&deg&$p_E^{\rm (shift)}$&\text{$n_\sigma^{\rm (shift)}$}& $B$        \\ \hline
$\{$3,7$\}$ & 0.1993505(5) & 0.1999(8) & 0.7 & 2& 0.1941(2) & -26 & 6&0.19318(9) &$-69$& 0.081(5)\\
$\{$7,3$\}$ & 0.5305246(8) & 0.5320(5) & 3.0 & 2& 0.52109(8) & -120 &4&0.52042(5)&$-200$& 0.087(2)\\
$\{$3,8$\}$ & 0.1601555(2) & 0.160(2) & -0.08& 3&0.1519(4) & -21 & 7& 0.1524(1)&$-78$& 0.26(1)\\
$\{$8,3$\}$ & 0.5136441(4) & 0.513(2) & -0.3 & 3&0.5032(2) & -52 & 6&0.5026(1) &$-110$& 0.32(3)\\
$\{$4,5$\}$ & 0.2689195(3) & 0.2695(6) & 1.0 & 2&0.2581(2) & -54 & 5& 0.2547(2)  &$-71$& 0.306(8) \\
$\{$5,4$\}$ & 0.3512228(3) & 0.3519(7) & 1.0 & 2&0.3415(4) & -24 & 2&0.3412(4)   &$-25$& 0.18(9) \\
$\{$4,6$\}$ & 0.20714787(9) & 0.2076(2) & 2.3& 1&0.19564(4) & -290 &3& 0.1949(3) &$-41$& -0.08(3) \\
$\{$6,4$\}$ & 0.3389049(2) & 0.3395(1) & 6.0 & 1&0.3271(4) & -29 & 2&0.3275(3)   &$-38$& 0.14(4) \\
$\{$5,5$\}$ & 0.25416087(3) & 0.2545(7) & 0.5& 2&0.2437(4) & -26 &5& 0.2453(2)   &$-44$& 0.88(6) \\
 \hline
 $\{$4,4$\}$ & 1/2 & {\bf 0.4897(3)} & 
  -34 &3& 0.500004(2) & 2 &6& 0.499992(6) & -1.3 & 0.003(1) \\
 $\{\infty,5\}$ & 1/4 & 0.2500(2) & 0&2 & \text{--} & \text{--} & \text{--} & \text{--} & \text{--} &\text{--}
\end{tabular}%
}
\caption{Critical $p$ values found for graph families characterized by
  Schl\"afli symbols $\{f,d\}$, where $f=\infty$ stands for random
  graphs of degree $d$.  Here $p_c$ is the percolation threshold
  (using invasion percolation data from
  Ref.~\onlinecite{Mertens-Moore-2017}, or exact values where known),
  $p_c^{(2/3)}$ is the percolation threshold using random-graph-like
  cluster size scaling (see Fig.~\ref{fig:ns23}), $p_E^{(C)}$ is the
  cycle erasure pseudothreshold obtained using the crossing point
  method (see Fig.~\ref{fig:PE-4-4} and top plots in
  Figs.~\ref{fig:pE-hyp55} and \ref{fig:pE-hyp73}), and
  $p_E^{\rm (shift)}$ is the same from the crossing point with an
  additional shift as in the bottom plots in Figs.~\ref{fig:pE-hyp55}
  and \ref{fig:pE-hyp73}, with $B$ the shift coefficient.  Numbers in
  the parenthesis indicate the standard deviation $\sigma$ in the
  units of the last significant digit, so that, e.g.,
  $0.14(4)\equiv 0.14\pm0.04$.  Values of
  $n_\sigma\equiv (p-p_{\rm c})/\sigma$ give the ``number of sigmas''
  for the deviation of the corresponding critical value found from the
  invasion percolation or exact threshold value if available.  Numbers
  in columns labeled ``deg'' give the degrees of the polynomials used
  to interpolate the data; polynomials of the same degrees were used
  to obtain $p_E^{(C)}$ and $p_E^{\rm (shift)}$.}
\label{tab:pc}
\end{table*}

\begin{figure}[htbp]
  \centering
  \includegraphics[width=0.6\columnwidth]{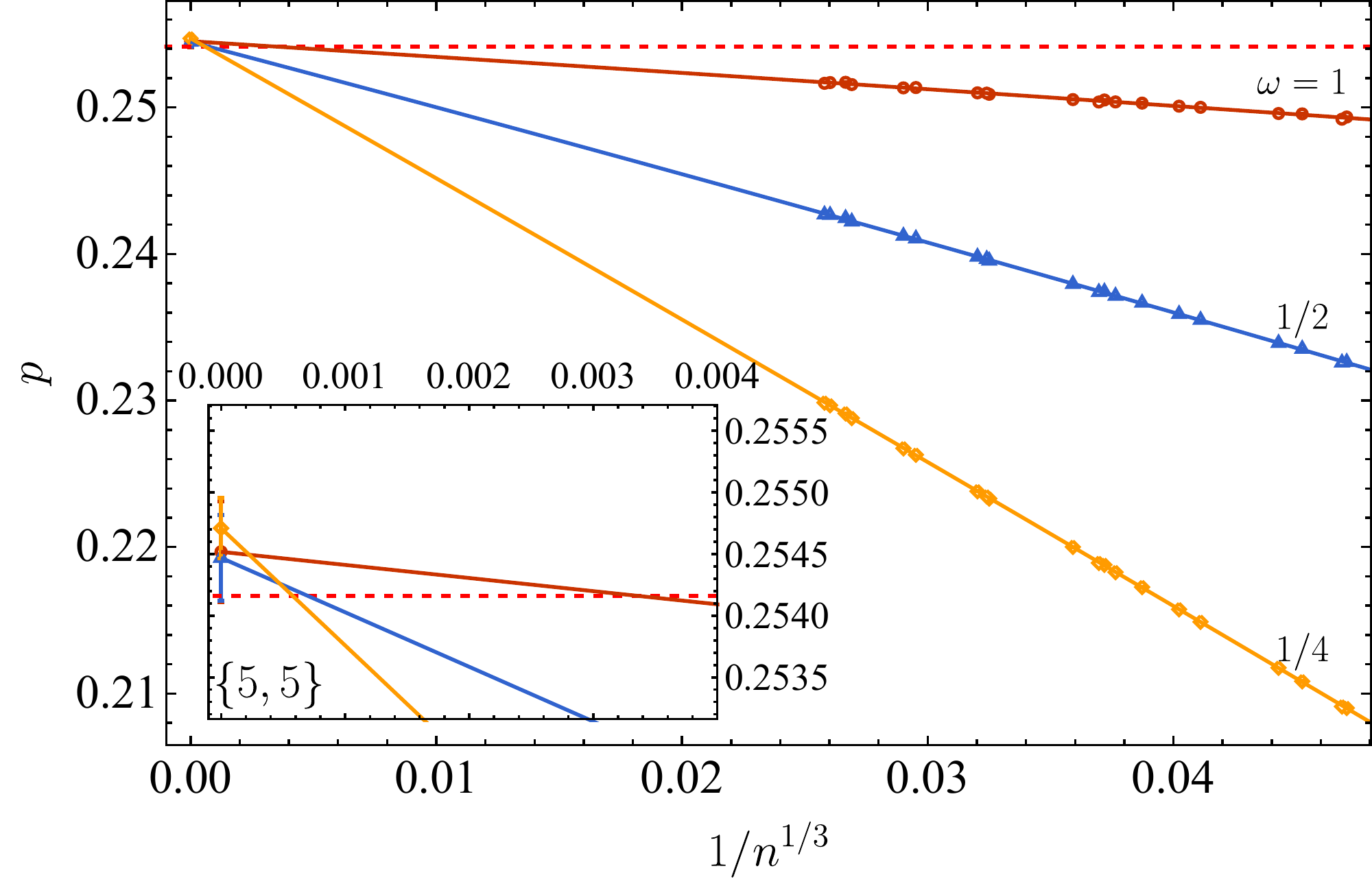}\\
  \includegraphics[width=0.6\columnwidth]{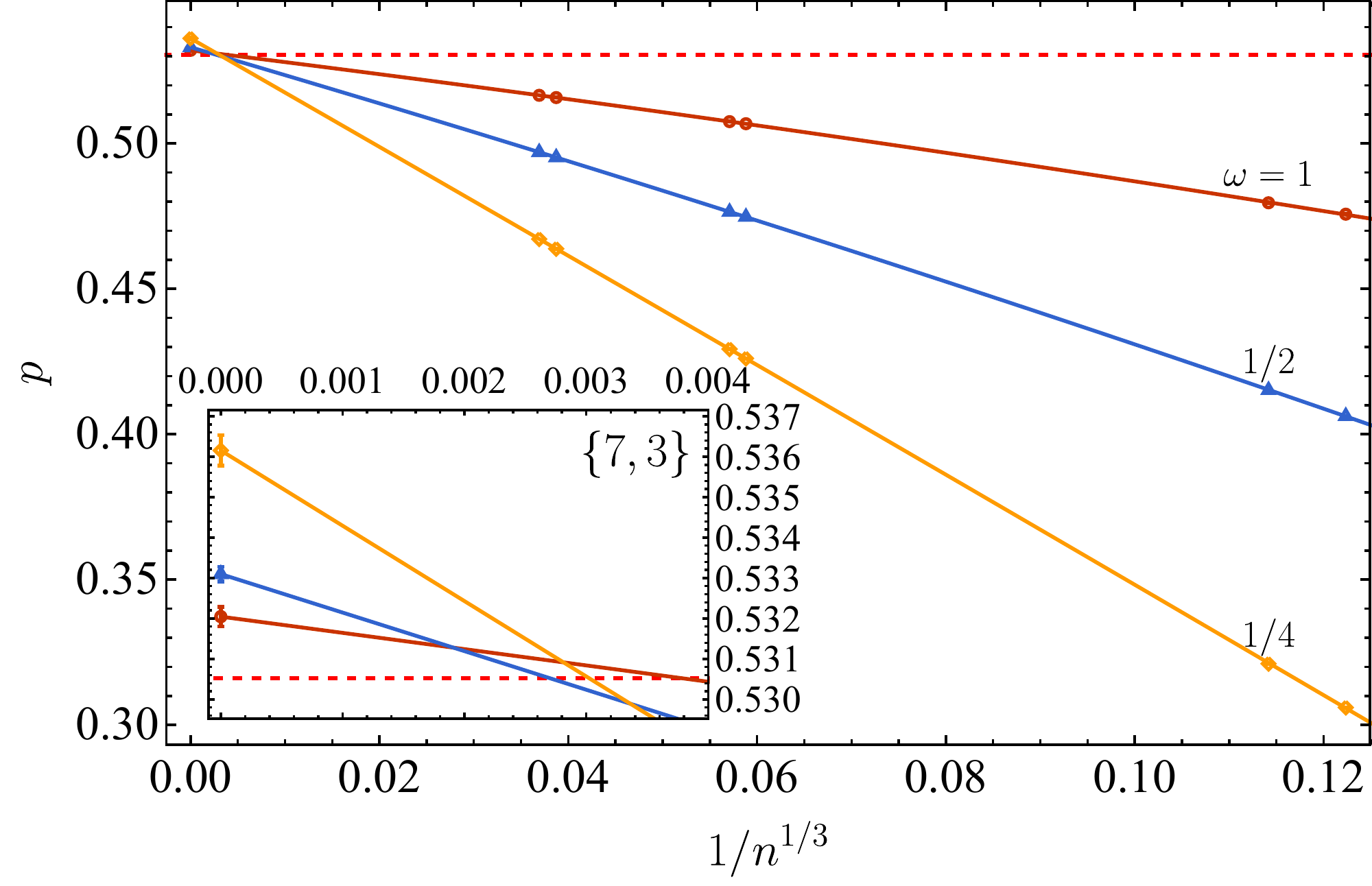}\\
  \includegraphics[width=0.6\columnwidth]{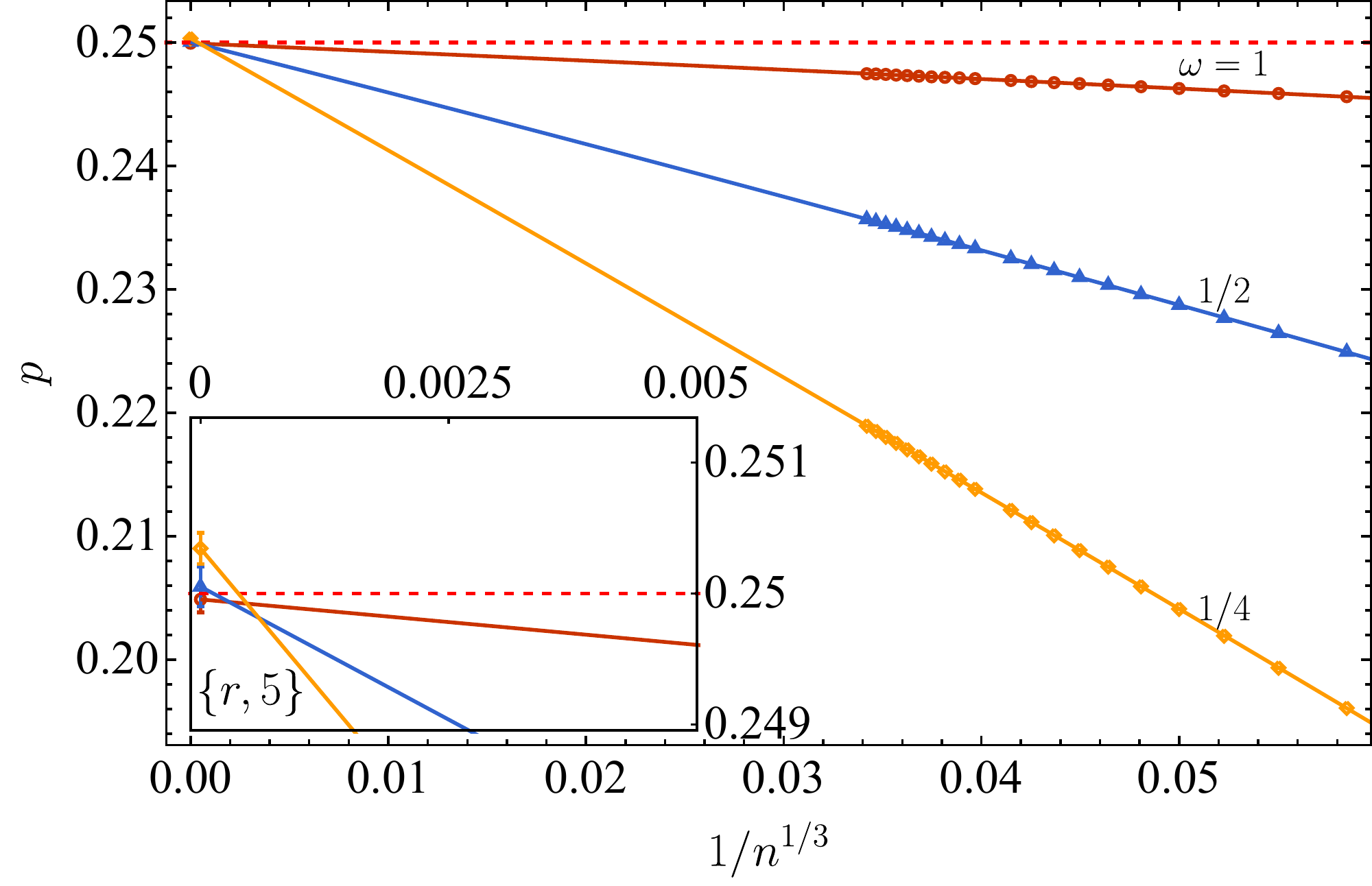}
  \caption{(Color online) Using the random-graph-like scaling for locating the
    percolation transition  for hyperbolic graphs
    in the $\{5,5\}$ (top) and $\{7,3\}$ (middle) families, and for
    degree-$5$ random graphs (bottom).  Values of the open bond
    probability $p$ where the expected size of the largest cluster
    equals $\omega n^{2/3}$ are plotted as a function of $n^{-1/3}$,
    for values of $\omega$ as indicated.  Here $n$ is the number of
    edges in the graph.  The lines intersect close to the percolation
    transition point, as indicated by horizontal dashed lines.
  }
\label{fig:ns23}
\end{figure}

The obtained critical values $p_c^{(2/3)}$, $p_E^{(C)}$, and
$p_E^{\rm (shift)}$ for different graph families are summarized in
Tab.~\ref{tab:pc}, where they are compared with the corresponding
percolation thresholds from Ref.~\onlinecite{Mertens-Moore-2017}
obtained from invasion percolation simulations, or exact values where
available.  Numerical data indicates that the erasure
(pseudo)threshold is substantially below $p_c$ for hyperbolic graphs
with logarithmic distance scaling, with the variation of the ratio
$\ln n/d$ having a significant effect on the quality of the crossing
point.  In contrast, for graphs from the $\{4,4\}$ family where
$d\propto n^{1/2}$, the cycle erasure (pseudo)threshold is very close
to the bulk percolation threshold, as generally expected from
Refs.~\onlinecite{Stace-Barrett-Doherty-2009,Fujii-Tokunaga-2012} and
Theorem \ref{th:pE-logB}.

Our results also indicate that for expander graphs, most accurate
results for percolation transition critical point are obtained using
the random-graph-like scaling, see Fig.~\ref{fig:ns23}, although this
technique is not at all applicable when the limiting graph is a tiling
of the euclidean plane.  Detailed comparison of the performance of
different extrapolation methods for percolation transition critical
point for various amenable and non-amenable graph families will be
published elsewhere.

\section{Conclusions}
\label{sec:conclusions}

In this work we focused on critical points associated with
homology-changing percolation transitions in a sequence of finite
graphs weakly convergent to an infinite graph $\mathcal{H}$, a
covering graph of the graphs in the sequence.  We also quantified the
relation between these critical points and edge percolation threshold
on $\mathcal{H}$.

The position of the homological $1$-cycle erasure threshold $p_E^0$ is
governed by the scaling of the homological distance $d$ with $\log n$,
where $d$ is the size of a smallest non-trivial cycle and $n$ is the
graph size (number of edges).  Generally,
$p_E^0\le p_{\rm c}(\mathcal{H})$, where the equality is reached for
superlogarithmic distance scaling, while $p_E^0=0$ is expected for
sublogarithmic distance scaling.  In the case of logarithmic distance
scaling where the quantity $d/\ln n$ remains bounded away from 0 and
from infinity, the cycle erasure threshold $p_E^0$ remains strictly
positive as long as $\mathcal{H}$ is a bounded-degree graph, and we
expect $p_E^0$ to be strictly below $p_{\rm c}$.

For an amenable graph $\mathcal{H}$ with a finite isoperimetric
dimension, an easy upper bound on the distance can be constructed by
considering a ball with the radius equal to the injectivity radius,
giving a power-law scaling of the distance with $n$.  Generically, we
expect that a sequence of covering maps with superlogarithmic distance
scaling can be constructed when such a graph is quasitransitive,
resulting in $p_E^0=p_{\rm c}(\mathcal{H})$.  In particular, this is
the case for any periodic lattice in dimension $D>1$, since covering
maps can be constructed by using periodic boundary conditions along
each axis.

On the other hand, logarithmic scaling of the distance is the most one
can expect when $\mathcal{H}$ is non-amenable.  For such a graph the
uniqueness threshold is expected\cite{Haggstrom-Jonasson-2006} to be
strictly higher than the percolation threshold,
$\Delta p \equiv p_{\rm u}(\mathcal{H})-p_{\rm c}(\mathcal{H})>0$,
which gives a non-trivial upper bound for the asymptotic rate,
$R\le\Delta p$, where $R>0$ corresponds to an extensive scaling of the
homology rank associated with non-trivial $1$-cycles.  For any graph
sequence with $R>0$, we also introduced a pair of homological
thresholds $p_H^0$ and $p_H^1$, associated with the points where
asymptotic erasure rate (\ref{eq:mean-erased}) deviates from the
values at $p=0$ and $p=1$, respectively.  Generally,
$p_H^0\ge p_{\rm c}$; for planar transitive graphs we proved
$p_H^0=p_{\rm c}(\mathcal{H})$ and $p_H^1=p_{\rm u}(\mathcal{H})$.  We
conjecture this to be the case more generally.

A number of open questions remain.  First, related to the sequences of
finite graphs both weakly convergent to an infinite graph
$\mathcal{H}$, and covered by $\mathcal{H}$.  What are the properties
of $\mathcal{H}$ necessary for such a sequence to exist, in
particular, is it necessary that $\mathcal{H}$ be quasi-transitive?
Second, is it true that with a logarithmic distance scaling, the
strict inequality holds $p_E^0<p_c(\mathcal{H})$?

Finally, an important open question is to what extent present results can be
extended to other models, in particular, Ising and, more generally,
$q$-state Potts model on various graphs.  Indeed, successful decoding
probability in qubit quantum LDPC codes can be mapped to ratios of
partition functions of associated random-bond Ising
models\cite{Dennis-Kitaev-Landahl-Preskill-2002,%
  Kovalev-Pryadko-SG-2015,Jiang-Kovalev-Dumer-Pryadko-2018}.  In the
clean (no-disorder) limit, these can be rewritten in terms of
Fortuin-Kasteleyn (FK) random-cluster models.  For such a model with
$q\ge1$, Hutchcroft \cite{Hutchcroft-2019} has recently proved the
exponential decay of cluster size distribution in the subscritical
regime.  In particular, this could help fixing the location of the
boundary of the decodable region for certain families of graph-based
quantum CSS codes in the weak-noise limit. 

\acknowledgements
\medskip\noindent\textbf{Acknowledgment:}
This work was supported in part by the NSF Division of Physics via
grant No.\ 1820939.

\appendix

\section*{Appendix: The proofs}

\subsection{Proof of Lemma \ref{th:dZ-bnd}}

\dZbnd*

\begin{proof}
  Let ${C}\subset {\cal E}_{\cal G}$ be a non-trivial cycle of weight
  $d_Z$, and $v\in\mathcal{V}_{\cal G}$ a vertex on $C$.  Let
  $v'\in{\cal V}$ be a vertex from the fiber of $v$, then the ball
  $\mathcal{B}\equiv \mathcal{B}(v',r_f;\mathcal{H})$ is mapped
  one-to-one by $f$.  Since $\mathcal{C}$ is non-trivial, it must
  contain at least one edge outside of the image of $\mathcal{B}$.
  Since $\mathcal{C}$ is also a minimum-weight non-trivial cycle, it
  must be self-avoiding, i.e., it should contain two edge-disjoint
  paths connecting $v$ to the boundary of the image of $\mathcal{B}$.
  Necessarily, $d_Z> 2r_f$.

  Conversely, consider a ball $\mathcal{B}_1$ of radius $r_f+1$ which
  covers a non-trivial cycle ${C}_1$ on $\mathcal{G}$ of weight $w$,
  the shortest cycle among those covered by $\mathcal{B}_1$.  At most
  two vertices of ${C}_1$ are at the distance $r_f+1$ from $u$
  (otherwise a shorter cycle could be constructed), which gives
  $d_Z\le w\le 2r_f+3$.
\end{proof}
\subsection{Proof of Lemma \ref{th:dX-bnd}}
\dXbnd*

\begin{proof}
  The statement is trivial if $r_f< \omega$, since $d_X\ge 1$ by
  definition.  Assume $r_f\ge\omega$, so that any generator of the
  cycle group on ${\cal H}$ be mapped one-to-one.  Thus, any (finite)
  cycle on ${\cal H}$ is mapped to a homologically trivial cycle,
  where we assume that the symmetric set difference ``$\oplus$'' be
  used when an edge is encountered in the image more than once.
  Consequently, a lift of a walk cycling around a simple non-trivial
  cycle ${C}$ on $\mathcal{G}$ cannot be closed; instead, it must be a
  portion of a semi-infinite self-avoiding path on ${\cal H}$.
  Respectively, for any edge $e_0\in{C}$ and its lift
  $e_0'\in {\cal E}$ such that $f(e_0')=e_0$, we denote
  ${C}'\equiv C'(C,e_0')\ni e_0'$, the \emph{extended lift} of $C$,
  the union of lifts of all walks on ${C}$ starting with $e_0'$ and
  $e_0$, respectively; ${C}'$ is an infinite self-avoiding path.

  Now, take a binary vector $b$ with $\wgt(b)=d_X$ such that
  $B\equiv \supp(b)\subset {\cal E}_{\cal G}$ be a minimum-weight
  non-trivial co-cycle on ${\cal G}$.  Then, it must be irreducible,
  which implies that $B$ must be \emph{cycle-connected}, i.e., for any
  pair of edges $e_{\rm i}\neq e_{\rm f}$ in $B$, it should also
  contain a connecting edge sequence
  $\mathcal{S}=(e_1=e_{\rm i}, e_2, \ldots,e_{m-1}, e_m=e_{\rm
    f})\subseteq B$, with any pair of neighboring edges sharing an
  image of a basis cycle on $\mathcal{H}$.  Given such a sequence of
  length $m$, the conventional graph distance between any pair of
  vertices from the union $e_{\rm i}\cup e_{\rm f}$ must be strictly
  smaller than $\omega m$.

  To prove the contrary, let us assume that $d_X\le r_f/\omega$.
  Then, a minimum-weight co-cycle $B\subset {\cal E}_{\cal G}$ must
  have a diameter strictly smaller than $r_f$, i.e., there be a ball
  ${\cal B}_r\subset {\cal G}$ of radius $r\le r_f$ such that
  $B\subset {\cal B}_r$.  Indeed, with $\wgt (B)=d_X$, any connecting
  sequence contains at most $m=d_X$ edges, which implies the
  conventional distance between any pair of vertices on $B$ smaller
  than $d_X\omega\le r_f$.  This implies that any lift $B'$ of $B$
  should be mapped one-to-one by $f$.

  To finish the proof, let $C\subset {\cal E}_{\cal G}$ be an
  irreducible cycle conjugate to $B$, i.e., the corresponding binary
  vectors satisfy $bc^T=1$, which implies the existence of an edge
  $e_0\in B\cap C$.  Irreducibility of $C$ implies that it must be a
  simple cycle on ${\cal G}$.  Given $e_0'$ such that $f(e_0')=e_0$,
  let $B'\ni e_0'$ be a lift of $B$ and $C'=C'(C,e_0')$ an extended
  lift of $C$, an infinite self-avoiding path on ${\cal H}$.  Since
  $B'$ is mapped one-to-one by $f$, it has odd-weigh intersection with
  $C'$ and even-weight intersection with any basis cycle on
  ${\cal H}$.  Respectively, $B'$ must have an odd-weight intersection
  with any deformation $\bar C'\equiv C'\oplus M$ of $C'$, where
  $M\subset {\cal C}_H$ is a finite cycle on $H$.  Thus, $B'$ is a
  finite-size cut splitting ${\cal H}$ into infinite portions, which
  can not be the case since $\mathcal{H}$ is assumed one-ended.
\end{proof}

\subsection{Proof of Lemma \ref{th:flat-sub}}

\FlatSub*

\begin{proof}
  Consider a graph $\mathcal{G}'$ obtained from $\mathcal{G}$ by
  removing the ball
  $\mathcal{B}\equiv \mathcal{B}(v,r_f;\mathcal{G})$.  Construct a
  connected graph $\mathcal{G}''$ from a union of $\mathcal{B}$ and
  spanning trees of every connected component of $\mathcal{G}'$, by
  sequentially adding bridge bonds connecting individual components so
  that no new cycles are introduced.  Such a subgraph contains all
  vertices of $\mathcal{G}$ and can be lifted to $\mathcal{H}$
  starting with $v'$; let
  $\mathcal{V}_f\subset \mathcal{V}_{\cal H}$ be the corresponding
  vertex set.  By construction, $f$ acts one-to-one on ${\cal V}_f$.
  It is also easy to check that $\mathcal{H}_f$, the subgraph of
  $\mathcal{H}$ induced by $\mathcal{V}_f$, be connected.
\end{proof}

\subsection{Proof of Theorem \ref{th:pE-easy}}

\thErasurePercol*

\begin{proof}\label{proof:thErasurePercol}
  If $p_{\rm c}(\mathcal{H})=1$, the statement of the theorem is trivial.
  In the following, assume $p_{\rm c}(\mathcal{H})<1$ and take $p$ such that
  $p_{\rm c}(\mathcal{H})<p<1$.  For some $t\in\mathbb{N}$, a chosen
  $v'\in \mathcal{V}_{\cal H}$ and $v\equiv f_t(v')$, we connect
  percolation on $\mathcal{H}$ and on $\mathcal{G}_t$ using a set-up
  similar to invasion percolation\cite{Wilkinson-Willemsen-1983}.
  Namely, we start with single-site zeroth generation clusters
  $\mathcal{K}_v^{(0)}=\{v\}\subset \mathcal{V}_t$ and
  $\mathcal{K}_{v'}^{(0)}=\{v'\}\subset \mathcal{V}_{\mathcal{H}}$,
  with no edges labeled open or closed.  Given a generation-$j$
  cluster $\mathcal{K}_{v'}^{(j)}\subset \mathcal{V}_{\cal H}$,
  every previously unlabeled edge adjacent to a vertex in
  $\mathcal{K}_{v'}^{(j)}$ is labeled open with independent
  probability $p$ and otherwise closed.  The next generation cluster
  $\mathcal{K}_{v'}^{(j+1)}$ is formed by adding any vertices
  connected to those in $\mathcal{K}_{v'}^{(j)}$ by newly open edges.
  Let us denote by $P_j(p;\mathcal{H})$ the probability that the
  process can be continued after step $j$, i.e., there be one or more
  unlabeled edges incident on the $j$-th generation cluster.  Clearly,
  $P_0(p;\mathcal{H})=1$ and $P_j(p;\mathcal{H})$ is strictly
  decreasing as a function of $j$, with
  $\lim_{j\to\infty}P_j(p;\mathcal{H})=\theta_{v'}(p;\mathcal{H})$.

  Let us now look at thus constructed percolation process on
  $\mathcal{G}_t$.  As long as the image of no vertex
  connected to $\mathcal{K}_{v'}^{(j)}$ by a so far unlabeled edge
  coincides with the image of a vertex in $\mathcal{K}_{v'}^{(j)}$ (we
  call such a cluster ``flat''), we can use the map $f_t$ to make the
  labels on $\mathcal{G}_t$ match those on $\mathcal{H}$.  Clearly,
  all clusters are flat for $j<r_t$, the injectivity radius; for such
  $j$ the probabilities that the percolation process may be continued
  match exactly on the two graphs,
  $P_j(p;\mathcal{G}_t)=P_j(p;\mathcal{H})$.  On the other hand,
  $P_j(p;\mathcal{G}_t)=0$ for $j\ge |\mathcal{V}_t|$.  The
  percolation processes necessarily decouple whenever a cluster
  $\mathcal{K}_{v'}^{(j)}$ ceases to be flat, i.e., there be an
  unlabeled edge on $\mathcal{G}_t$ connecting a pair of vertices in
  $\mathcal{K}_{v}^{(j)}\subset\mathcal{V}_t$.  Given such a cluster,
  we can assign the remaining unlabeled edges on $\mathcal{G}_t$ all
  at once; the resulting open subgraph of $\mathcal{G}_t$ contains a
  homologically non-trivial cycle with probability greater than or
  equal to $p$.  At the same time, the cluster
  $\mathcal{K}_{v'}^{(j)}\subset \mathcal{V}_{\cal H}$ is removed from
  the percolation process on $\mathcal{H}$.  Since it is not certain
  that a descendant of a given cluster be infinite, we get the lower
  bound
  \begin{equation}
  \mathbb{P}(\text{$\mathcal{K}_v$ contains a non-trivial cycle})\ge
  p\,\theta_{v'}(p;\mathcal{H}),\label{eq:prob-non-triv}
\end{equation}
which is positive for any $p>p_{\rm c}$,
  thus $p_E^0\le p_{\rm c}$.
\end{proof}

\subsection{Proof of Theorem \ref{th:pE-logA}}

\pElogA*

\begin{proof}
  To set up independent erasure events, cut ${\cal G}_t$ into
  non-overlapping regions, images of non-overlapping balls on
  ${\cal H}$ of radius ${\rho}_t\equiv 1+\lfloor d_{Zt}/2\rfloor$.
  Given the maximum graph degree $\Delta_{\rm max}$, we can cut out at
  least
  $$N_t\ge |V_t|/|\mathcal{B}_0(2{\rho}_t,{\cal G}_t)|> |V_t|/\Delta_{\rm max}^{2+d_{Zt}}$$
  such balls.  By transitivity of $\mathcal{G}_t$ and Lemma
  \ref{th:dX-bnd}, each ball contains a homologically non-trivial
  cycle of length $d_{Zt}$, which is open with probability
  $P_1\ge p^{d_{Zt}}$.  Now, probability that a homology is covered in
  \emph{none} of the $N_t$ balls can be upper bounded as
  \begin{eqnarray}\nonumber
    P_{\rm none}
    &=&
        [1-P_1]^{N_t}\le [1-p^{d_{Zt}}]^{N_t}\le    \exp(-N_t p^{d_{Zt}}) \\
    &<& \exp\left(- |V_t|/\Delta_{\rm max}^2\, e^{-d_{Zt}\left(\left|\ln p\right|
          +\ln \Delta_{\rm max}\right)}\right),\;
          \label{eq:upper-bnd-none}
  \end{eqnarray}
  which is guaranteed to converge to zero for any $p>0$ since
  $d_{Zt}$ scales sublogarithmically with $|V_t|$.  (Notice that
  $|V_t|\ge 2n_t/\Delta_{\rm max}$ by a version of the hand-shaking
  lemma.)
\end{proof}

Notice that the requirement of transitivity for the graphs
$\mathcal{G}_{Zt}$ can be relaxed a bit, namely, by assuming that the
number of vertex classes [defined by distinct vertex orbits connected
by elements of $\aut(\mathcal{G}_t)$] remains uniformly bounded for
the graphs $\mathcal{G}_t$.  In that case, the balls need to be taken
of radius $\rho_t=\lfloor d_{Zt}/2\rfloor+m$, where $m$ is the maximum
number of vertex classes.  The proof is completed with the following
lemma:
\begin{lemma} Consider a connected graph ${\cal H}$, with $m\ge 1$
  vertex classes.  Any ball of radius $m$ contains representative(s)
  of all classes.
\end{lemma}
\begin{proof} Consider a class connectivity graph ${\cal G}$
corresponding to ${\cal H}=({\cal V},{\cal E})$, with $m$ vertices
(one per class) and an edge between two vertices if ${\cal H}$
contains an edge between a pair of vertices in these classes.
Necessarily, ${\cal G}$ is connected.  Further, given a vertex $v\in
{\cal V}$, any spanning tree on ${\cal G}$ can be lifted to a tree on
${\cal H}$ that contains $v$; such a tree contains a representative
from every vertex class.  Further, the diameter of the tree cannot
exceed $m$; such a tree is contained in a ball ${\cal B}(v,m;{\cal
H})$.  The proof is complete since the choice of $v$ is arbitrary.
\end{proof}

\subsection{Proof of Theorem \ref{th:pE-logB}}

\pElogB*

\begin{proof}
  Only a cluster with $s\ge d_{Zt}> r_t$ vertices can cover a
  homology.  For a graph ${\cal G}$, let $ M_v(s;{\cal G})$ denote the
  probability that vertex $v$ is in an open cluster with exactly $s$
  vertices on $[{\cal G}]_p$.  On the quasi-transitive graph
  ${\cal H}$, this probability has an exponential bound,
  $M_v(s;{\cal H})< M(s)\equiv e^{-\gamma(p)s}$, for some $\gamma(p)$
  non-zero in the subcritical region, $\gamma(p)>0$ for
  $p<p_{\rm c}$\cite{Antunovic-Veselic-2008}.  Note also
  $\sum_{s\ge 1}M_v(s;{\cal G})=1$ on any finite graph; below
  percolation threshold this is also true for infinite graphs.  Also,
  for any $v\in{\cal V}_t$, finding a cluster of size $s\le r_t$
  attached to $v$ on ${\cal G}_t$ has the same probability as that
  attached to a vertex $v'(v)$ from the fiber of $v$ on $\mathcal{H}$.
  Use the union bound for the probability of finding a cluster of size
  $r_t+1$ or larger on ${\cal G}_t$,
  \begin{eqnarray}
    \nonumber P_\mathrm{one}
    &\le & \sum_{v\in {\cal V}_t}\sum_{s>r_t} s^{-1}M_v(s;{\cal G}_t)
    \\ \nonumber
    &<&
        \sum_{v\in {\cal V}_t}\sum_{s> r_t} M_v(s,{\cal G}_t)\\ \nonumber
    &=&\sum_{v\in{\cal V}_t}\left(1-\sum_{1\le s\le r_t} M_v(s;{\cal G}_t)\right)
    \\ \nonumber
    &=&\sum_{v\in{\cal V}_t}
        \left(1-\sum_{1\le s\le r_t} M_{v'(v)}(s;{\cal H})\right)
    \\ \nonumber
    &=& \sum_{v\in{\cal V}_t}\sum_{s> r_t} M_{v'(v)}(s;{\cal H})\\
    &<& |{\cal V}_t|\sum_{s>r_t} e^{-\gamma(p)s}= {|V_t| e^{-\gamma(p) r_t}\over
        e^{\gamma(p)}-1},
            \label{eq:upper-bnd-one}
  \end{eqnarray}
  which goes to zero with $t\to\infty$ whenever $\gamma(p)>0$ since
  $r_t$ is assumed to be superlogarithmic in $n_t\ge |{\cal V}_t|-1$.
  This proves $p_E^0\ge p_{\rm c}$; the statement of the Theorem is
  obtained with the help of Theorem \ref{th:pE-easy}.
\end{proof}
\subsection{Proof or Corollary \ref{th:pE-planar}}
\pEplanar*

\begin{proof}
  Since $\widetilde{\cal H}$ is quasitransitive, it has a finite
  maximum degree, which is the maximum size of a face of ${\cal H}$.
  Thus, with injectivity radius large enough, $f_t$ must be invertible
  on the union of any face and its adjacent faces on ${\cal H}$.  This
  guarantees that (with $t$ sufficiently large, $t>t_0$), ${\cal G}_t$
  be locally planar, so that we can construct the locally planar dual
  graph $\widetilde{\cal G}_t$ whose cover is $\widetilde{\cal H}$.
  Further, for any open edge configuration, the ranks of the homology
  groups on the open subgraph of ${\cal G}_t$ and on the closed
  subgraph of $\widetilde{\cal G}_t$ add to $k_t$, the number of
  inequivalent homologically non-trivial cycles on ${\cal G}_t$
  [Eq.~(\ref{eq:rank-H1-restricted-duality})].  Thus, the two erasure
  thresholds are simply interchanged by duality,
  $\tilde{p}_E^1=1-p_E^0$ and $\tilde{p}_E^0=1-p_E^1$, so that the
  inequality $p_E^1\ge 1-p_{\rm c}(\widetilde{\cal H})$ follows
  immediately from Theorem \ref{th:pE-easy}.

  The identities in (\textbf{ii}) and (\textbf{iii}) similarly follow
  from Theorems \ref{th:pE-logA} and \ref{th:pE-logB} with the help of
  Lemmas \ref{th:dZ-bnd} and \ref{th:dX-bnd} which guarantee that the
  injectivity radii on the sequence of mutually dual graphs scale
  simultaneously in a sub-logarithmic, logarithmic, or
  superlogarithmic fashion.
\end{proof}

\subsection{Proof of Theorem \ref{th:pH-diff}}

\pHdiff*

\begin{proof}
  Let $k_t=\rank H_1(f_t)$ be the number of non-trivial
  independent cycles on $\mathcal{G}_t$.  Consider any open edge
  configuration on $\mathcal{G}_t$, with homology rank $k_t'\le k_t$,
  and another edge configuration obtained by removing some open edges,
  with homology rank $k_t''\le k_t'$.  Such a change in homology
  requires removing at least $\Delta k_t=k_t'-k_t''$ open edges.
  Considering these as random edge configurations at $p'>p_H^1$ and
  $p''<p_H^0$, averaging, and dividing by the total number of edges
  $n_t$, we obtain
  $$p'-p''\ge \mathbf{R}_E(t,p')-\mathbf{R}_E(t,p'');
  $$ in the limit $t\to\infty$ this
  becomes $p'-p''\ge R$.  Taking infimum over $p'>p_H^1$ and supremum
  over $p''<p_H^0$, we obtain the claimed inequality.
\end{proof}

\subsection{Proof of Theorem \ref{th:pH-easy-bounds}}
\pHeasybounds*
\begin{proof}
  Take $p>p_H^0$, then the limit in Eq.~(\ref{eq:pH0-def}) is either
  strictly positive or does not exist.  In either case, since terms in
  the sequence are bounded, $\mathbf{R}_E(t,p)<1$, the superior limit
  $f_p\equiv \limsup_{t\to\infty}\mathbf{R}_E(t,p)$ exists and is
  strictly positive, $f_p>0$ at $p>p_H^0$.  This implies the existence
  of a convergent subsequence, e.g., specified by an increasing
  sequence of indices $(t_j)_{j\in\mathbb{N}}$ such that
  $\lim_{j\to\infty}\mathbf{R}_E(t_j,p)=f_p$.

  Because of the existence of the limit, whenever $f_p>0$, for any
  $\epsilon>0$ and a sufficiently large $j$, clusters covering
  homologically non-trivial cycles are expected to occupy at least
  $(f_p-\epsilon)n_{t}$ edges, where $t\equiv t_j$.  Thus, if we
  choose $\epsilon=f_p/2$, a cluster $K_v\subset [{\cal G}_t]_p$
  connected to a randomly chosen vertex $v\in{\cal V}_t$ covers a
  homologically non-trivial cycle with probability
  $P_{\text{non-triv}}\ge f_p n_t/(2|{\cal V}_t|)$.  Using a map like
  that in the proof of Theorem \ref{th:pE-easy}, at sufficiently large
  $t$, a cluster covering a non-trivial cycle on $[{\cal G}_t]_p$
  corresponds to an infinite cluster on $[{\cal H}]_p$, which gives
  $$\theta_v(p)\ge \lim_{j\to\infty}f_p n_{t_j}/(2|{\cal V}_{t_j}|)\ge f_p/2>0,$$
  thus
  $p>p_{\rm c}({\cal H})$.
\end{proof}

\subsection{Proof of Theorem \ref{th:pH-planar}}

\PHplanar*

\begin{proof}
  As in the proof of Corollary \ref{th:pE-planar}, at sufficiently
  large $t$ the graph ${\cal G}_t=({\cal V}_t,{\cal E}_t)$ is
  necessarily locally planar, which implies the existence of the
  corresponding dual graph
  $\widetilde{\cal G}_t=(\widetilde{\cal V}_t,\widetilde{\cal E}_t)$, with the
  dual-graph sequence weakly convergent to the dual infinite graph
  $\widetilde{\cal H}$.

  The proof relies on the relation\cite{Delfosse-Zemor-2012} between
  the expected homology rank of the open subgraph and the expected
  inverse cluster sizes on an open subgraph of $\mathcal{H}$ and a
  closed subgraph of $\widetilde{\cal H}$.  While the argument goes
  back to the work of Sykes and Essam\cite{Sykes-Essam-1964}, we give
  a complete derivation here.  Consider a configuration of open/closed
  edges on $[{\cal G}_t]_p$ with $E'\le n_t\equiv |{\cal E}_t|$ open
  edges, $K'$ clusters, and the cycle group of rank
  $C'=C_{\rm triv}+k'$, where $k'$ is the number of non-trivial basis
  cycles.  According to Euler's theorem, $K'=|{\cal V}_t|-E'+C'$.  On
  the other hand, duality matches any simple trivial cycle on
  ${\cal G}_t$ to a cut on the dual graph $\widetilde{\cal G}_t$,
  which gives $C_{\rm triv}=\tilde K'-1$, with $\tilde{K}'$ being the
  number of clusters on the dual graph in the dual edge configuration,
  with open and closed edges interchanged.  This gives
  $$
  k'=K'-\tilde{K}'+E'-|{\cal V}_t|+1.
  $$
  Taking the average over the edge configurations on $[{\cal G}_t]_p$
  we obtain for $k_p^{(t)}\equiv \Ep\biglb(\rank H_1(f_t,p)\bigrb)$,
  \begin{eqnarray*}
    k_p^{(t)}
    &=&\sum_{v\in{\cal V}_t}\kappa_v(p;{\cal G}_t)
        -
        \sum_{v\in\widetilde{\cal V}_t}\kappa_v(\bar p;\widetilde{\cal G}_t)
              +pn_t-|{\cal V}_t|+1.
  \end{eqnarray*}
  Here $ \kappa_v(p;{\cal G}_t)\equiv \Ep\left(|{K}_v|^{-1}\right)$ is
  the expected inverse size of a cluster containing vertex $v$ on
  $[{\cal G}_t]_p$, and $ \kappa_v(\bar p;\widetilde{\cal G}_t)$ is the
  corresponding quantity on the dual graph, averaged over the dual
  edge configurations, which is equivalent to $\bar p=1-p$.
  Introducing the corresponding vertex-average quantities, e.g.,
  $\kappa(p;{\cal G}_t)\equiv |{\cal V}_t|^{-1}\sum_{v\in{\cal
      V}_t}\kappa_v(p;{\cal G}_t)$, we get%
  $$
  k_p^{(t)}
  =|{\cal V}_t|\kappa(p;{\cal G}_t)-|\widetilde{\cal V}_t|\kappa(1-p;\widetilde{\cal G}_t)+pn_t-|{\cal V}_t|+1.
  $$
  To obtain the asymptotic erasure rate (\ref{eq:mean-erased}) we
  divide the obtained result by $n_t$ and notice that very large
  clusters give no contribution to the total while (at sufficiently
  large $t$) any finite cluster on $[{\cal G}_t]_p$ has the same
  probability as an equivalent cluster on $[{\cal H}]_p$.  Further,
  assuming the transitive graphs ${\cal H}$ and $\widetilde{\cal H}$
  of degrees $d$ and $f$, respectively, the graphs ${\cal G}_t$ and
  $\widetilde{\cal G}_t$ respectively have the same degrees, and the
  hand-shaking lemma gives $|{\cal V}_t|=2n_t/d$,
  $|\widetilde{\cal V}_t|=2n_t/f$.  This proves both the existence and
  the value of the following limit at any $p$,
  \begin{eqnarray}
    \nonumber
    \mathbf{R}_E(p)
    &\equiv& \lim_{t\to\infty}  \mathbf{R}_E(t,p)\\
    &=& {2\over  d}\kappa(p;{\cal H})
        - {2\over  f}\kappa(\bar p;\widetilde{\cal H})  +p- {2\over  d}.
        \label{eq:RE-transitive}
  \end{eqnarray}
  Finally, we notice that for non-amenable transitive graphs
  ${\cal H}$ and $\widetilde{\cal H}$, the quantities
  $\kappa(p;{\cal H})$ and $\kappa(\bar p;\widetilde{\cal H})$ are
  analytic functions of $p$ in the vicinity of any $p\in(0,1)$ such
  that $p\neq p_{\rm c}({\cal H})$ and
  $\bar p\neq p_{\rm c}(\widetilde{\cal H})$,
  respectively\cite{Antunovic-Veselic-2008,Hermon-Hutchcroft-2019}.
  Thus, the r.h.s.\ of Eq.~(\ref{eq:RE-transitive}) is an analytic
  function of $p$ for
  $$p\in(0,1)\setminus\{p_{\rm c}({\cal H}),1-p_{\rm c}(\widetilde{\cal H})\},$$ where
  $1-p_{\rm c}(\widetilde{\cal H})=p_{\rm u}({\cal H})>p_{\rm c}({\cal H})$.  On the other hand,
  $\mathbf{R}_E(p)$ cannot be analytic in the lower and upper
  homological tresholds $0<p_H^0<p_H^1<1$, which gives the two
  equalities.
\end{proof}

\bibliography{lpp,qc_all,more_qc,percol,spin,linalg,ldpc,grav,teach}

\end{document}